\documentclass[12pt]{article}       
\usepackage{color}
\usepackage[utf8]{inputenc}
\usepackage[T1]{fontenc}
\usepackage{babel}
\usepackage{mathtools}
\usepackage{amsfonts}
\usepackage{amssymb}
\usepackage{graphicx}
\usepackage{dsfont}
\usepackage{epsfig}
\usepackage{epstopdf}
\usepackage{ae}
\usepackage{amsmath}
\usepackage{amsthm}
\usepackage{thmtools}
\usepackage{endnotes}
\usepackage{setspace}
\usepackage{amsthm}
\usepackage{appendix}
\usepackage{float}
\usepackage{multirow}
\usepackage{chngpage} 
\usepackage{bbm}
\usepackage{subcaption}
\usepackage[round]{natbib}
\usepackage{enumerate}
\usepackage{xcolor}
\usepackage{url}
\usepackage{booktabs}

\newcommand{\F}{\mathcal{F}}

\newcommand{\R}{\mathbb{R}}
\newcommand{\Q}{\mathbb{Q}}
\renewcommand{\P}{\mathbb{P}}
\newcommand{\D}{\mathcal{D}}

\renewcommand{\[}{\left[}

\renewcommand{\tilde}{\widetilde}
\newcommand{\1}{\mathds{1}}

\newtheorem{theo}{Theorem}[section]

\newtheorem{lem}{Lemma}[section]
\newtheorem{co}{Corollary}[section]
\newtheorem{assumption}{Assumption}[section]
\theoremstyle{definition}
\newtheorem{rem}{Remark}[section]

\renewcommand{\geq}{\geqslant}
\renewcommand{\leq}{\leqslant}
\renewcommand{\ge}{\geqslant}
\renewcommand{\le}{\leqslant}
\renewcommand{\bar}{\overline}

\renewcommand\thmcontinues[1]{Continued}

\renewcommand{\d}{\textrm{d}}
\newcommand{\Var}{\text{Var}}

\DeclarePairedDelimiter{\norm}{\lVert}{\rVert}
\DeclareMathOperator{\Cov}{Cov}
\DeclareMathOperator{\E}{\mathbb{E}}
\DeclareMathOperator{\arginf}{\arg\inf}

\usepackage{natbib}

\begin{document}
	
	\author{ 
    \textsc{Nicole B\"{a}uerle}\footnote{Department of Mathematics, Karlsruhe Institute of Technology, Germany, nicole.baeuerle@kit.edu} \footnote{The author gratefully acknowledges financial support from the Deutsche Forschungsgemeinschaft (DFG), Project-ID 509303834, FOR 5583 ``Asset Allocation and Asset Pricing under Regulatory Uncertainty''}
    \and
    \textsc{Anne MacKay}\footnote{Department of Mathematics, Universit\'{e} de Sherbrooke, Canada, anne.mackay@usherbrooke.ca} \footnote{The author gratefully acknowledges financial support from NSERC Discovery Grant number RGPIN-2024-05794 and from the International Excellence Fellowship, supported by the University of Excellence concept of the Karlsruhe Institute of Technology (KIT).}
	}
	\title{Mean-Variance Optimization in Ambiguous Financial Markets with Learning}
	\date{\today}
	
	\maketitle

    \begin{abstract}
        We consider a continuous time investment problem in a multi-asset Black-Scholes market with the following features: The assets' drifts are not known and constitute a source of model ambiguity. However, there is a prior distribution (knowledge)  on the possible drifts. Our investor is ambiguity averse and wants to maximize a mean-variance criterion for the terminal wealth where ambiguity aversion is incorporated in a smooth way. We consider here the criterion introduced in \cite{maccheroni_alpha_2013} where the variance is decomposed and each part is weighted differently to account for different levels of market risk and model ambiguity aversion. We use a novel approach to find the optimal dynamic investment strategy within the class of all adapted strategies which allow for learning. We also present a number of numerical results which help to understand how the model parameters affect the optimal investment strategy. In general it turns out that ambiguity averse investors invest less in the risky assets. 
    \end{abstract}

\vspace{0.3cm}
\begin{minipage}{14cm}
{\small
\begin{description}
\item[\rm \textsc{ Key words}]
{\small Portfolio Optimization, Bayesian Problem, Mean Variance,  Smooth Ambiguity, Duality Theory, Fredholm Integral equation}
\item[\rm \textsc{MSC classifications}] {\small 91G10, 93E11, 93E20, 45B05  }
\end{description}
}
\end{minipage}
\section{Introduction}
The mean-variance criterion is commonly used to identify optimal asset allocation strategies. It can be traced back to \cite{MV52} and - less well-known - to de Finetti in 1940 (see \cite{Pressacco2007}). It provides a tractable framework for balancing expected return against risk and has first been treated in a static framework. It can also be seen as a good approximation to utility maximization, not only in the context of a Gaussian distribution, \cite{markowitz_meanvariance_2014,zakamouline2009generalisation} or as an approximation of the entropic risk measure as a risk-sensitive criterion, \cite{bielecki2003economic,bauerle_markov_2024}.
Enventually, the mean-variance criterion has also been applied in a continuous-time trading framework, in particular when the stock price evolution is given by the Black Scholes model, see \cite{korn1997optimal,bajeux1998dynamic} for some early references.

The main aim of this paper is to consider a continuous-time portfolio allocation problem with the mean-variance criterion under ambiguity, i.e.\ the investor is unsure about the correct probabilistic financial market model. More precisely, we consider the Black-Scholes financial market where the drift of the assets is not known. This is a realistic feature since the drift of the stocks is notoriously difficult to estimate. We take a Bayesian perspective and assume the existence of a prior distribution over the possible drifts. When the investor is ambiguity neutral and just integrates over all possible drifts w.r.t.\ their distribution, this is simply the well-known Bayesian case, which has been considered in the literature. We treat a similar setting but consider an ambiguity averse investor. To represent this situation, a suitable criterion has been introduced in \cite{maccheroni_alpha_2013} as a generalization of the classical Arrow-Pratt approximation of the certainty equivalent. In our case, it can also be seen as a generalization of the variance decomposition formula. Suppose that $X$ is a random wealth which depends on unknown (random) parameters $\Theta$. Applying the variance decomposition formula to the mean variance criterion with a penalty weight $\alpha<0$ yields
\begin{align}\nonumber
      \E[X] +\alpha \Var(X)= \E[X] +\alpha \E[\Var(X|\Theta)]+\alpha \Var(\E[X|\Theta]).
\end{align}
In an ambiguity averse setting, we weight the two parts of the variance differently. More precisely, let  $\alpha,\beta <0$.  We consider the objective function
\begin{align}\nonumber
      \E[X] +\alpha \Var(X)= \E[X] +\alpha \E[\Var(X|\Theta)]+\beta \Var(\E[X|\Theta]).
\end{align}
The part  $\Var(\E[X|\Theta])$ can be seen as an ambiguity premium and we thus assume that $\beta<\alpha<0.$ Simple static portfolio problems with this criterion have been solved in \cite{maccheroni_alpha_2013}. In this paper  we show --- for the first time --- how to solve the dynamic asset allocation problem with this criterion. We solve the problem over all admissible, adapted strategies that take statistical learning of the parameters into account. We use a novel approach to achieve this. In the case of a discrete prior, the solution hinges on the solution of a linear equation. In the case of a continuous prior, we have to solve an integral equation to obtain the solution. \\

{\em Related literature. }
The {\em classical continuous-time mean variance problem} has been solved in \cite{korn1997optimal,bajeux1998dynamic,zhou_continuous-time_2000,zhou2003markowitz}. The latter paper also includes a regime-switching feature of the model parameters. The papers differ in their precise formulation of the optimization criterion; the mean-variance problem can either be formulated as a constrained problem or with a weighted criterion. Both formulations are essentially equivalent, since the weighted criterion is the Lagrangian of the constrained problem, and both lead to a stochastic linear-quadratic (LQ) control problem. In \cite{li2002dynamic}, the authors consider 
a mean-variance portfolio allocation problems in continuous-time under the constraint that short-selling of stocks is prohibited. The problem is formulated as
a stochastic optimal LQ-problem with constraints. \cite{bielecki2005continuous}
study a continuous-time mean-variance portfolio selection problem  where all the market coefficients are random and the wealth process under any admissible trading strategy is not allowed to be below zero at any time.
More general financial market models are considered in \cite{xia2006markowitz}, where an incomplete semimartingale model is treated and solved with duality methods. Different notions of mean-variance optimality, in particular an instantaneous version, can be found in \cite{maclean_mean-variance_2011,pedersen2017optimal}. The problem of finding time-consistent optimal asset allocation strategies within restricted classes of policies is considered in \cite{basak2010dynamic,bjork2014mean}. Let us finally mention \cite{lim2004quadratic} who treats a closely related problem, namely a quadratic hedging problem and uses backward stochastic differential equations to solve the appearing LQ-problems.

A subset of the literature on continuous-time mean variance optimal strategies is concerned with the {\em Bayesian mean-variance problem}, meaning that some parameters of the model, e.g.\ the drifts, are unknown, and information on these unknown parameters is only given by a prior distribution. The probability measure of the Bayesian problem is given by the product measure of the prior and the measure for stock price movement, e.g. the law of the Wiener process. The general approach here is to derive the conditional distribution of the unknown parameters given market observations, also known as filter process, to be included in the state space. A solution in the setting of general stochastic parameters can be found in \cite{xiong_meanvariance_2007}, and in the case of an unknown drift, in \cite{xiong2021mean}. In \cite{xiong2020mean} the authors consider partially observed point processes for the stock price movement. The extension to a hidden Markovian regime-switching Black–Scholes model is treated in \cite{elliott2010mean,yang2015mean}. In \cite{liu2016optimal}, inflation is added to the model and influences the asset prices and can be observed directly or through a noisy observation.

Finally let us mention some related papers which deal with {\em portfolio optimization under ambiguity.} Ambiguity as an important aspect of decision making has already been described in \cite{ellsberg1961risk}. More precisely, people tend to be ambiguity averse when they can choose between a known model and an unknown model; they
prefer to invest in the known model even when the odds seem to be the same in both cases.
\cite{dimmock2015estimating} develop a method to estimate multiple prior models of decision-making under ambiguity. In a representative sample of the U.S. population, they measure ambiguity attitudes in the gain and loss domains.
In \cite{brenner2018asset}, the authors develop an empirical methodology for measuring the degree of ambiguity and for assessing attitudes toward ambiguity from market data. \cite{easley2009ambiguity} investigate the implications of ambiguity aversion for performance and regulation of markets. In their model, agents’ decision making may incorporate both risk and ambiguity, and the authors demonstrate that nonparticipation arises from the rational decision by some traders to avoid ambiguity. Keeping these findings in mind, several ways to mathematically model ambiguity aversion have been proposed. For example, in \cite{garlappi2007portfolio} the authors consider a static mean-variance problem where the drift is unknown and has to be estimated. Ambiguity aversion is modeled by a max-min optimization where the mean-variance objective is maximized under the minimal value for the drift which is within the estimation error. Such a robust approach is also treated in \cite{ismail2019robust} for a continuous-time investment problem where the model uncertainty affects the covariance matrix of multiple risky assets. The problem is formulated into a min–max mean‐variance problem over a set of nondominated probability measures that is solved by a McKean–Vlasov dynamic programming approach. In \cite{van_staden_surprising_2021}, the authors discuss the robustness of the optimal mean-variance portfolio to model misspecification. 

Since the robust approach is sometimes considered to be too pessimistic and because in that framework, the likelihood of the different models have no impact, \cite{klibanoff2005smooth} introduced the concept of smooth ambiguity. 
This idea has been used in \cite{maccheroni_alpha_2013} to deriving a tractable mean-variance model adjusted for ambiguity. The authors also solve a static asset allocation problem in this framework. The same criterion has also been investigated in \cite{zhang2024robust}, Sec.\ 3. However, the author restricts to solving this problem in the class of equilibrium strategies which implies that the unknown drift parameters cannot be learned, cp.\ also \cite{balter2021time} where the same is done for a  power utility investment problem under ambiguity. On the other hand, \cite{bauerle_optimal_2024} were the first to solve a smooth ambiguity asset allocation problem within a Black-Scholes market with unknown drift parameters, under the assumption that parameters can be learned. This work is performed under a pair of power utilities for market and model risk. Here, instead, we consider the mean-variance criterion which cannot be treated in the same way due to its non-linearity features.

Finally, it is also important to mention that the recent advancements in machine and reinforcement learning have brought a new perspective on the mean-variance portfolio optimization problem. 
Recently, \cite{huang2024mean} propose an adaptation of a policy gradient-based actor-critic algorithm to the mean-variance optimization problem. 
Their approach is data-driven and does not rely on a specific model, other than the market being described by diffusion processes.
In this sense, their method also includes model ambiguity; however, ambiguity aversion is not explicitly included in their objective function.
For an overview of the literature on reinforcement learning for the mean-variance optimization problem, we refer the reader to the introduction of \cite{huang2024mean}.

The paper is organized as follows. In the next section, we introduce the multi-asset Black-Scholes financial market model with unknown drift and recall important facts about filtering, which are used in the sequel. Throughout, we work with a complete financial market. In Section 3, we first recall the solution of the mean-variance problem in the classical Bayes case, which corresponds to an ambiguity neutral investor. We solve the problem using the duality or martingale approach, which means that we first determine the optimal terminal wealth and in a second step, derive its replicating strategy. In the second part of Section \ref{sec:main} we investigate the model with ambiguity. We first treat the case where the prior is concentrated on a finite number of points. We use again the duality approach and exploit the fact that the variance can be written as an optimization problem. The optimal terminal wealth can be derived in terms of the solution of a specific linear equation. The optimal investment strategy then follows in an established way. The same approach can be used in the case of a continuous prior where however, the linear equation is replaced by an integral equation, which is much more demanding from a theoretical point of view. We also derive the complete solution in case of a single stock normal prior for the drift. Finally, in Section \ref{sec:experiments}, we discuss and interpret numerical examples, and Section \ref{sec:conclusions} concludes.

\section{Financial model}\label{sec:financial_market}

We consider a filtered probability space $(\Omega,\F,(\F_t),\P)$ and fix a finite time horizon $T > 0$. We also let $W = (W_1(t),\ldots,W_d(t))_{t\geq 0}$ be a standard $d$-dimensional Wiener process and denote by $\F^W = (\F^W_t)_{t\geq 0}$ its augmented filtration.

\subsection{Financial market}

The financial market contains $d$ risky assets with price process $S=(S_1(t),\ldots,S_d(t))_{0\leq t \leq T}$ and a bond. For simplicity, we assume that the bond price is constant at 1. For each asset $i \in \{1,\ldots,d\}$ we assume that $S_i(0)=1$ and the dynamics of its price $S_i$ are given by
\begin{align*}
	\d S_i(t) &= S_i(t)\Big(\mu_i \, \d t + \sum_{j=1}^d \sigma_{ij} \d W_j(t)\Big),
\end{align*}
where $\sigma_{ij} \in \mathbb R_+$ for $j \in \{1,\ldots,d\}$ and $\sigma = (\sigma_{ij})_{i,j=1}^d$ is invertible.
In general, we will assume that the drift $\mu = (\mu_1,\ldots,\mu_d)^\top$ is unknown to the investor and therefore modeled as a random variable taking values in $\mathbb R^d$. Equivalently, we may assume that the market price of risk $\theta := \sigma^{-1}\mu$ is unknown. We denote the corresponding random variable by $\Theta$ and assume here that the (prior) distribution $\Q_0(A)= \P(\Theta \in A)$, $A \in \mathcal B(\R^d)$ is known and $\Theta$ is independent of the Wiener process $W$. More precisely our probability measure $\P$ may be written as $\P = \Q_0\otimes \P^W$ where $\Q_0$ is the prior distribution of $\Theta$ and $\P^W$ the distribution of the Wiener process. In what follows, we denote by $\D \subset \R^d$ the (measurable) support of the prior. Later we will write $\P_\theta(\cdot) = \P(\cdot |\Theta=\theta)$ for the conditional distribution. We assume that 
$$ \int_{\D} \|\theta\| \Q_0(d\theta)<\infty.$$

Let us define $Y(t) = W(t) + \Theta t$, so that for each $i \in \{1,\ldots,d\}$, 
\begin{equation}\label{eq:assetdynamics}
	\d S_i(t) = S_i(t) \left(\sum_{j=1}^d \sigma_{ij} \d Y_j(t)\right), \quad t \in [0,T].
\end{equation}
We further let $\mathcal F^Y = (\mathcal F^Y(t))_{0\leq t \leq T}$ denote the filtration generated by $Y,$ that is
\begin{align*}
	\mathcal F^Y (t) = \sigma(Y(s), 0 \leq s \leq t),
\end{align*}
and remark that this filtration coincides with the one which is generated by the stock price processes $S$. 

\subsection{Admissible investment strategies}
We let $\pi = (\pi_1,\ldots,\pi_d)$ be a $d$-dimensional stochastic process that describes the investment strategy; $\pi_i(t)$ represents the amount invested in the $i^{\text{th}}$ asset at time $t \in [0,T]$.
Since the price of the bond is constant, and because we only consider self-financing strategies, we can describe an admissible investment strategy by the positions taken in the risky assets.
We also restrict ourselves to strategies $\pi$ that are $\F^Y$-adapted; the investor adjusts her strategy after observing asset prices.
The wealth process resulting from the application of an investment strategy $\pi$ is denoted $X^\pi = (X^\pi(t))_{0\leq t \leq T}$ and given by
\begin{align}\label{eq:dX_pi}
	\d X^\pi(t) = \sum_{i=1}^d \pi_i(t) \frac{\d S_i(t)}{S_i(t)} = 
	\pi(t) \sigma \d Y(t),
\end{align}
with $X^\pi(0) = x_0$, the initial wealth invested in the portfolio.
We denote by $\Pi(x_0)$ the set of $\F^Y$-adapted, self-financing strategies with initial wealth $x_0 > 0$ that satisfy
$$ \int_0^T \E \|\pi(t)\|^2 dt <\infty.$$

\subsection{Facts about partially observable financial markets}
In what follows we summarize well-known, important facts about this partially observable financial market, see e.g.\ \cite{jouini_bayesian_2001}.
In order to solve the stochastic control problem we have to consider the filter process, i.e.\ the conditional probability of $\Theta$ being in specific states, given our observations. It holds that  the posterior distribution of $\Theta$ given the observations $\F^Y_t$ 
\begin{align}
	\Q_t(A) \coloneqq \P(\Theta \in A | \F^Y_t), \qquad A \in \mathcal B(\D),
\end{align} 
is given by
\begin{align}\label{eq:filter}
	\Q_t(A) =  \E[\1_A(\Theta)  | \F_t^Y]
			= \left\{ \begin{array}{ll}
				\Q_0(A), &t=0,\\[0.2cm]
				\frac{\int_A \exp(\theta^\top y - \frac 12 \norm{\theta}^2 t) \Q_0(\d \theta)}{\int_{\D} \exp(\theta^\top y - \frac 12 \norm{\theta}^2 t) \Q_0(\d \theta)} \mid_{y = Y(t)}, &t>0.
			\end{array}\right.
\end{align}
In what follows it is convenient to define
\begin{align*}
  & F(t,y) = \int_{\D} \exp\big(\theta^\top y - \frac 12 \norm{\theta}^2 t\big) \Q_0(\d \theta).
\end{align*}
Thus, the time-$t$ mean vector of this conditional distribution, denoted by $\hat \theta(t)$, is the Bayes estimator of $\theta$ on $[0, t]$ and is given by
\begin{align}
	\hat \theta(t) = \int_{\D} \theta \, \Q_t(\d \theta)
	= \begin{cases}
		\int_{\D} \theta \, \Q_0(\d \theta), &t=0,\\
		\frac{\nabla F}{F}(t,Y(t)), &t>0.
	\end{cases}
\end{align}

We further have to compute risk-neutral prices of claims. Thus, we  define the usual state-price density  $Z=(Z(t))_{t\geq 0}$  by
\begin{align}\label{eq:Z}
	Z(t) &= \exp\left(\Theta^\top W(t) + \frac 12 \norm{\Theta}^2 t\right)= \exp\left(\Theta^\top Y(t) - \frac 12 \norm{\Theta}^2 t\right)
\end{align}
and   $\Lambda=(\Lambda(t))_{t\geq 0}$ by $\Lambda(t) = Z^{-1}(t)$. 
Since $\Lambda$ is a $(\F,\P)$-martingale, we can define a risk-neutral measure $\tilde \P$ by
\begin{align*}
	\tilde \P (A) = \E[\Lambda(T)\1_A], \qquad  \, A \in \F_T.
\end{align*}
However, here we have to work with the smaller filtration $\F^Y.$
Thus, we define the $(\F^Y, \tilde \P)$-martingale $\hat Z = (\hat Z(t))_{t \geq 0}$ by
\begin{align}
	\hat Z(t) = \tilde \E[Z(T)| \F^Y_t]
	= \begin{cases}
		1, & t=0,\\
		F(t,Y(t)), &t > 0.
	\end{cases}
\end{align}

We can finally define the innovation process $(N(t))_{t\geq0}$ by
\begin{align}
	N(t) = Y(t) - \int_0^t \hat\theta(s) \,\d s,
\end{align}
which is a $(\F^Y,\P)$ Wiener process.

Further, when we introduce $(\hat \Lambda(t))_{t \geq 0}$, with $\hat \Lambda(t) = \hat Z^{-1}(t)$ for $t \geq 0$ and assume $\pi\in\Pi(x_0)$ is an admissible investment strategy with corresponding wealth $X^\pi$,  then,
	\begin{equation}
		d(\hat \Lambda (t) X^\pi(t)) = \hat \Lambda (t)
		\left[\sigma^\top \pi(t) - \hat\theta(t) X^\pi(t)\right]^\top \d N(t),
	\end{equation}
	so that $(\hat \Lambda (t) X^\pi(t))_{t \geq 0}$ is a  $(\F^Y,\P)$-local martingale.

Setting $\pi\equiv0$ it can be shown that 
	\begin{align*}
		\d\hat\Lambda(t) = -\hat\Lambda(t) \hat\theta^\top(t) \d N(t),
	\end{align*}
	so that $(\hat\Lambda(t))_{t \geq 0}$ is a $(\F^Y, \P)$-martingale.

\begin{lem}
Assume $\pi\in\Pi(x_0)$ is an admissible investment strategy, then we have the following $(\F^Y_t)$-adapted representation for the wealth portfolio
\begin{align*}
	dX^\pi(t) = \pi(t)(\sigma \hat \theta(t) \, \d t + \sigma \d N(t)), \qquad X(0)=x_0.
\end{align*}

\end{lem}

It is important to note that in this setting, the market is complete, in the sense that for any contingent claim $H \in \mathbb H \coloneqq L^2(\Omega, \F^Y_T,\P)$, there exists an $\F^Y$-adapted, square-integrable $d$-dimensional process $(\xi(t))_{t\geq 0}$ satisfying
\begin{equation*}
	H\hat\Lambda(T) = \E[H \hat \Lambda(T)] + \int_0^T \xi(t)^\top \, \d N(t). 
\end{equation*}
Thus, the problem with partial information can be cast into the classical complete market setting with adjusted drift process.

\section{Mean-variance optimization with uncertainty}\label{sec:main}

To study the optimal mean-variance portfolio with parameter uncertainty, we adopt the setting of \cite{maccheroni_alpha_2013}.
To do so, it helps to understand uncertainty as a two-stage random experiment: the outcome of the first experiment is $\theta$, and the second yields a path of the Wiener process from 0 to $T$. The first experiment can be considered as model ambiguity, the second as the 'usual' uncertainty. We want to use a mean-variance criterion for our portfolio optimization problem and at the same time distinguish between the two sources of randomness. Since the studies of \cite{ellsberg1961risk} we know that most people are risk-averse, but also ambiguity averse. Thus, it makes sense to model both separately and study their effect on optimal investment decisions.  In the framework of mean-variance such an approach has been proposed by \cite{maccheroni_alpha_2013}. The authors study in detail the properties and interpretations of the objective, but only solve a static portfolio problem in this setting. In our paper we solve an adapted portfolio problem with learning using the objective in \cite{maccheroni_alpha_2013}

In order to introduce the objective, we need some further notation. Recall that the underlying probability measure is given by $\P = \Q_0\otimes \P^W.$ Thus $\P_\theta(\cdot) = \P(\cdot |\Theta=\theta)$ can be seen as a transition kernel. Thus, for $X \in L^2(\Omega,\mathcal F^Y_T,\P)$, we let
\begin{align*}
  \E[X]=  \E_{ \P} [X] = \int\int X \, \d\P_\theta \d\Q_0,
    \qquad
    \E_{\Q_0} [X] =& \int X \,  \d\Q_0,
    \qquad \text{and} \qquad
    \E_\theta  [X] = \int X \,  \d\P_\theta.
\end{align*}
Now recall the classical variance decomposition formula which reads
\begin{align}\nonumber
    \Var(X) &= \E[\Var(X|\Theta)]+ \Var(\E[X|\Theta])
    = \E[\Var_\theta(X)]+ \Var(\E_\theta[X])
\end{align}
where we have
\begin{eqnarray*}
      \E[\Var_\theta(X)] &=& \int \left(\E_\theta [X^2] -(\E_\theta [X])^2\right) \Q_0(d\theta) = \E[X^2] - \int (\E_\theta [X])^2 \,\Q_0(d\theta) \\
      \Var(\E_\theta[X]) &=& \int (\E_\theta [X])^2 \,\Q_0(d\theta) -\Big( \int \E_\theta [X] \Q_0(d\theta) \Big)^2 = \int (\E_\theta [X])^2 \Q_0(d\theta) -\Big(  \E[X] \Big)^2 
\end{eqnarray*}


The objective function discussed in \cite{maccheroni_alpha_2013} is such that the two parts of the variance decomposition may have different weights. The motivation stems from a derivation of a certainty equivalent in a model with ambiguity. A part of $\Var(\E_\theta[X])$ can be seen as an ambiguity premium.
Thus, for $\alpha,\beta <0$, we consider the objective function

\begin{equation}
\E_{}[X] +\alpha \E[\Var(X|\Theta)]+\beta \Var(\E[X|\Theta])
\label{eq:CX3}
\end{equation}
which is a weighted criterion of expectation, expected variance of uncertainty and variance of expectation which can be seen as an ambiguity premium.

\subsection{The Bayesian case}\label{ssec:bayesian}

We first consider an ambiguity-neutral investor, that is, we want to solve the setting $\alpha=\beta<0.$ More precisely we consider
\begin{align}
	\label{eq:MV_bayesian}
    \sup_{\pi \in \Pi(x_0)} \E[X^\pi(T)]+\alpha \Var(X^\pi(T)).
\end{align}

Solving a Bayesian adaptive portfolio problem with a mean-variance objective function is similar to \cite{jouini_bayesian_2001}, who consider utility maximization instead. A solution (by different means) for the single stock setting can be found in \cite{xiong2021mean}.
In a more general setting, where $\mu$ and $\sigma$ can be adapted (and possibly unobserved) processes, the mean-variance optimization problem is solved in \cite{xiong_meanvariance_2007}. 
The ideas presented below constitute a new approach and summarize results from these works. 

First of all note the well-known fact that 
\begin{equation}\label{eq:variance_identity}
    \Var(X)= \inf_{b\in\R} \E[(X-b)^2].
\end{equation}
We use this to reformulate the objective function \eqref{eq:MV_bayesian} as follows  (cp. \cite{bauerle_time-consistency_2025}) where it is important to note that $\alpha<0$.
\begin{align}
	\label{eq:MV_bayesian_2}
    \sup_{\pi\in \Pi(x_0)} \sup_{b\in\R} \E\Big[X^\pi(T)+\alpha(X^\pi(T)-b)^2\Big].
\end{align}

Next we make use of the fact that we can always interchange $\sup$ and $\sup$.
Thus, the inner optimization problem  for a fixed $b\in\R$ is
\begin{align*}
	\sup_{\pi\in \Pi(x_0)}   \E\Big[X^\pi(T)+\alpha(X^\pi(T)-b)^2\Big],
	\qquad \text{s.t.} \quad 
	\E[\hat \Lambda(T) X^\pi(T)] = x_0.
\end{align*}
We solve this problem by finding first the optimal terminal wealth, which can in a second step be hedged by an investment strategy. This works, because the partial observation model can be reduced to a complete financial market model (see end of Section \ref{sec:financial_market}).
Finding the optimal terminal wealth can be done using the Lagrange function (in what follows we write $ \hat \Lambda$ instead of $ \hat \Lambda(T)$) with Lagrange multiplier $\lambda\in\R$, that is
\begin{align*}
   L(X,\lambda)&:=  \E\Big[X+\alpha(X-b)^2 +\lambda(\hat \Lambda X - x_0)\Big].
\end{align*}
Taking
the Fr\'echet derivative w.r.t. $X$ and setting it to zero yields
\begin{align*}
    1+2\alpha (X-b)+\lambda \hat \Lambda=0.
\end{align*}
Thus, we obtain
\begin{align*}
    X = b -\frac{1}{2\alpha}(1+\lambda \hat \Lambda).
\end{align*}
The Lagrange multiplier $\lambda$ and the parameter $b$ can be obtained from the two constraints 
\begin{itemize}
    \item[(i)] $\E[\hat\Lambda X]=x_0$, (budget constraint)
     \item[(ii)] $\E[ X]=b$. (expectation equality)
\end{itemize}
Note that (ii) immediately implies that $\lambda =-1.$ Plugging this into (i) (note that $\E[ \hat \Lambda]=1$) yields
\begin{align*}
b = x_0 - \frac{1}{2\alpha}\Big(\E[\hat\Lambda^2] -1\Big) \quad \mbox{and}\quad X= 
     x_0 - \frac{1}{2\alpha}\Big(\E[\hat\Lambda^2] -\hat \Lambda\Big)  \end{align*}
Finally, we know that the supremum of the outer optimization problem is attained at $b= \E[X].$ This equation can be solved for $b$ and yields the final solution for the optimal terminal wealth.
Below we summarize our findings, keeping in mind that $\hat \Lambda= F^{-1}(T,Y(T))$. The optimal investment strategy can be derived as in \cite{jouini_bayesian_2001}, Theorem 3.2. We summarize our results in the following theorem.

\begin{theo}\label{theo:Bayes}
    In the model of this section the optimal terminal wealth is given by
\begin{align*}
    X^B = x_0-\frac{1}{2\alpha}\Big(\E[\hat\Lambda^2(T)] - \hat\Lambda(T)\Big).
\end{align*}
It holds that
\begin{align*}
    \E[ X^B] = x_0 -\frac{1}{2\alpha} \Var(\hat\Lambda(T)), \quad \mbox{and} \quad \Var(X^B)= \frac{1}{4\alpha^2} \Var(\hat\Lambda(T)).
\end{align*}
The optimal investment strategy $\pi\in\Pi(x_0)$ is given by
\begin{align*}
    \pi^B(t) &= -(\sigma^\top)^{-1} \frac{1}{2\alpha} \int_{\R^d} \frac{\nabla F(T,Y(t)+z)}{F^2(T,Y(t)+z)} \varphi_{T-t}(z) dz,\quad t\in[0,T]
\end{align*}
where $\varphi_{T-t}(z)$ is the density of the  multivariate normal $\mathcal{N}_d(0,(T-t)I)$.
\end{theo}

\begin{rem}
    In the case $d=1$ (only one risky asset), the optimal terminal wealth coincides with the previous findings in the literature, e.g. \cite{xiong2021mean}, Lemma 3.1.
\end{rem}

\subsection{The model with ambiguity}\label{ssec:main_ambiguity}
As motivated at the beginning of this section we now consider the objective 
\begin{equation}
\E_{}[X] +\alpha \E[\Var(X|\Theta)]+\beta \Var(\E[X|\Theta])
\label{eq:CX2}
\end{equation}
with $\alpha,\beta<0.$ Recall that for $\alpha=\beta$ we are back in the previous setting. \cite{maccheroni_alpha_2013} called this a 'robust mean-variance' functional and derived it as a local approximation of a KMM preference functional. Now we consider the investment problem
\begin{align}
	\label{eq:MV_bayesian_A}
    \sup_{\pi \in \Pi(x_0)} \E[X^\pi(T)]+\alpha  \E[\Var(X^\pi(T)|\Theta)]+\beta \Var(\E[X^\pi(T)|\Theta]).
\end{align}
In order to solve this problem we follow an approach similar as before, but with some important variations. First note that we can re-write the objective function as
\begin{align}\nonumber
   & \E[X^\pi(T)]+\alpha  \E[\Var(X^\pi(T)|\Theta)]+\beta \Var(\E[X^\pi(T)|\Theta]) \\  \nonumber =
    & \E[X^\pi(T)]+\alpha \Var(X^\pi(T)) +(\beta-\alpha) \Var(\E[X^\pi(T)|\Theta]) \\ \label{eq:MVobjective_ambiguity} =& 
   \E[X^\pi(T)]+\beta  \Var(X^\pi(T))+(\alpha-\beta) \E[\Var(X^\pi(T)|\Theta)] 
\end{align}   

From the second line of this representation we see that in order to penalize the model ambiguity it is reasonable to assume that $(\beta - \alpha)$ is negative, since $\Var(\E[X^\pi(T)|\Theta])$ measures the variability of the expected terminal wealth across the unknown parameter. Thus, we assume that $\beta<\alpha<0$. Hence the value in \eqref{eq:MV_bayesian_A} is bounded from above by 
\begin{align*}
   & \sup_{\pi \in \Pi(x_0)} \E[X^\pi(T)]+\alpha  \E[\Var(X^\pi(T)|\Theta)]+\alpha \Var(\E[X^\pi(T)|\Theta])\\
   =&   \sup_{\pi \in \Pi(x_0)} \E[X^\pi(T)]+\alpha  \Var(X^\pi(T))<\infty.
\end{align*} 
The last expression is the objective function in the Bayesian case from which we already know that it is finite for all $\alpha<0.$ To solve the optimization problem, we consider \eqref{eq:MVobjective_ambiguity}. Observe that due to the definition of the conditional expectation as a projection in $L^2$ and the tower property of conditional expectation, we have
\begin{align*}
  \E[\Var(X^\pi(T)|\Theta)]= \E[\E[(X^\pi(T)-\E[X^\pi(T)|\Theta])^2|\Theta]] =   \inf_{b \in L^2(\D)} \E[(X^\pi(T)-b(\Theta))^2]
\end{align*}
where 
\begin{align*}
    \arginf_{b \in L^2(\D)} \E[(X^\pi(T)-b(\Theta))^2] = \E[X^\pi(T)|\Theta].
\end{align*}
Thus, our objective function is
\begin{align*}
    \sup_{\pi \in \Pi(x_0)} \inf_{b \in L^2(\D)}  \E[X^\pi(T)]+\beta  \Var(X^\pi(T))+(\alpha-\beta) \E[(X^\pi(T)-b(\Theta))^2].
\end{align*}
Next, it is possible to interchange $\sup$ and $\inf.$ The proof is deferred to the appendix.

\begin{lem}\label{lem:interchange}
    In the above setting we obtain
    \begin{align*}
& \sup_{\pi \in \Pi(x_0)} \inf_{b \in L^2(\D)}  \E[X^\pi(T)]+\beta  \Var(X^\pi(T))+(\alpha-\beta) \E[(X^\pi(T)-b(\Theta))^2]\\
= &   \inf_{b \in L^2(\D)}  \sup_{\pi \in \Pi(x_0)} \E[X^\pi(T)]+\beta  \Var(X^\pi(T))+(\alpha-\beta) \E[(X^\pi(T)-b(\Theta))^2]
   \end{align*}  
\end{lem}

Writing the second variance term also as an optimization problem we obtain the final form for our objective function,
\begin{align}\label{eq:MVobjective:Afinal}
    \inf_{b \in L^2(\D)}  \sup_{\bar b\in\R}\sup_{\pi \in \Pi(x_0)} \E[X^\pi(T)]+\beta  \E[(X^\pi(T)-\bar b)^2]+(\alpha-\beta) \E[(X^\pi(T)-b(\Theta))^2].
\end{align}
This representation will now be used to solve the problem. In the next subsections we consider both a discrete and a continuous prior.

\subsubsection{Discrete prior}

In order to explain the approach we first assume that the prior $\Q_0$ is concentrated on a finite number of $m$ points, i.e.\ $\D=\{\theta_1,\ldots ,\theta_m\}$ where $\theta_j\in\R^d.$




In this setup we obtain the solution presented in Theorem \ref{theo:Ambiguity} below.  
In order to formulate it, 
we define the linear equation
\begin{equation}\label{eq:bk}
    \left(I - \frac{\alpha - \beta}{\alpha} A\right) b = c,
\end{equation}
where $I$ is the $m$-dimensional identity matrix, $A \in \mathbb R^{m \times m}$ has elements
\begin{equation}\label{eq:aik}
    a_{ik} = \E_{\theta_i}\left[\Q_T(\theta_k)\right] - \Q_0(\theta_k)\E_{\theta_k}\left[\hat\Lambda(T)\right]
\end{equation}
and $c \in \mathbb R^m$ has elements
\begin{equation*}
    c_k = x_0 - \frac 1{2\alpha} \left(\E\left[\hat\Lambda^2(T)\right] - \E_{\theta_k}\left[\hat\Lambda(T)\right]\right).
\end{equation*}
We also write $A= A_1 - A_2$, where the elements of $A_1 \in \mathbb R^{m\times m}$ are given by the first term on the right-hand side of \eqref{eq:aik} 
and where $A_2 = \mathbf{1}_m \varsigma^\top$ is the outer product of $\mathbf{1}_m$ and $\varsigma \in \mathbb R^m$, with elements $\varsigma_k = \Q_0(\theta_k)\E_{\theta_k}\left[\hat\Lambda(T)\right]$.

The existence of a solution to \eqref{eq:bk} is discussed in Lemma \ref{lem:solution_bk} below.

\begin{lem}\label{lem:solution_bk}
    Assume that all the eigenvalues of $A_1$ are different from $\frac{\alpha}{\alpha - \beta}$. Then, there exists a solution $b = (b_1,\ldots,b_m)^T$ to \eqref{eq:bk} if
    \begin{equation}\label{eq:condition_bk}
        1 + \frac{\alpha-\beta}{\alpha} \varsigma^\top \left(I-\frac{\alpha-\beta}{\alpha} A_1\right)^{-1} \mathbf 1_m \neq 0.
    \end{equation}
\end{lem}

\begin{rem}
    A sufficient condition to ensure that all the eigenvalues of $A_1$ are different from $\frac{\alpha}{\alpha-\beta}$ is to have 
    \begin{equation}\label{eq:sufficient}
    2\alpha < \beta < \alpha,
    \end{equation}
    where the second inequality is already a condition of our model. 
    Indeed, since it holds that ${\sum_{k=1}^m \E_{\theta_i}[\Q_T(\theta_k)] = 1}$, the matrix  $A_1$ is a stochastic matrix, so all of its eigenvalues are contained in the interval $[-1,1]$.
    Then, when \eqref{eq:sufficient} is satisfied, $\left|\frac{\alpha}{\alpha-\beta}\right| > 1$ and thus cannot coincide with any of the eigenvalues of $A_1$.
    We emphasize however that this condition is not necessary for the existence of a solution to \eqref{eq:bk}, and may significantly restrict the range of attitudes towards risk and ambiguity.
\end{rem}

For the next theorem we use the abbreviation  $G_k(t,y):= F(t,y) L_k^{-1}(t,y)$ for $k=1,\ldots,m$ with
\begin{align*}
     L_k(t,y)&= \exp\big(\theta_k^\top y - \frac 12 \norm{\theta_k}^2 t\big), \quad k=1,\ldots,m.
\end{align*}
The proof is deferred to the appendix.

\begin{theo}\label{theo:Ambiguity}
    In the model of this section with ambiguity, and under the conditions of Lemma \ref{lem:solution_bk}, the optimal terminal wealth is given by
\begin{align}\nonumber
    X^*  = &\; x_0-\frac{1}{2\alpha}\Big(\E[\hat\Lambda^2(T)] - \hat\Lambda(T)\Big) \\ \nonumber
    & + \Big(\frac{\alpha-\beta}{\alpha}\Big) \sum_{k=1}^m b_k \Q_0(\theta_k) \left(\frac{L_k(T,Y(T))}{F(T,Y(T))}- \E_{\theta_k}[\hat\Lambda(T)] \right)\\ \nonumber
    = & \; X^B +
     \Big(\frac{\alpha-\beta}{\alpha}\Big) \sum_{k=1}^m b_k \Q_0(\theta_k) \left(\frac{L_k(T,Y(T))}{F(T,Y(T))}- \E_{\theta_k}[\hat\Lambda(T)] \right)\\ \label{eq:optwealth_ambiguity}
    = & \; X^B + X^H.
\end{align}
where $(b_1,\ldots,b_k)$ is the unique solution to \eqref{eq:bk} and $X^B$ is the optimal terminal wealth in the classical Bayesian case (Thm. \ref{theo:Bayes}).
The optimal investment strategy $\pi^*\in\Pi(x_0)$ for the risky assets  is given by
\begin{align}\nonumber
    \pi^*(t) =& -(\sigma^\top)^{-1} \frac{1}{2\alpha} \int_{\R^d} \frac{\nabla F(T,Y(t)+z)}{F^2(T,Y(t)+z)} \varphi_{T-t}(z) dz -\\ \nonumber
    &-\Big( \frac{\alpha-\beta}{\alpha}\Big) \sum_{k=1}^m b_k \Q_0(\theta_k) (\sigma^\top)^{-1} \int_{\R^d} \frac{\nabla G_k(T,Y(t)+z)}{G_k^2(T,Y(t)+z)} \varphi_{T-t}(z) dz \\ \nonumber
    = & \;\pi^B(t) -\Big( \frac{\alpha-\beta}{\alpha}\Big) \sum_{k=1}^m b_k \Q_0(\theta_k) (\sigma^\top)^{-1} \int_{\R^d} \frac{\nabla G_k(T,Y(t)+z)}{G_k^2(T,Y(t)+z)} \varphi_{T-t}(z) dz\\ \label{eq:optinvest_ambiguity}
    = &  \;\pi^B(t)  +\pi^H(t), \quad t\in[0,T]
\end{align}
where again $\pi^B$ is the optimal strategy in the classical Bayesian case (Thm. \ref{theo:Bayes}).
\end{theo}

We can see from this result that we can express both the optimal terminal wealth, as well as the optimal investment strategy in terms of the solution of the Bayesian case without ambiguity aversion. In particular the second term $\pi^H$ in \eqref{eq:optinvest_ambiguity} is the part which hedges against model ambiguity. It is possible to show that in the single stock case this hedging demand is always negative under a moderate assumption: We have to assume that the $b_k's$ (which are the conditional expectations $\E_{\theta_k}[X^*]$) are ordered in the same way as the market prices of risk  $\theta_k$. This was the case throughout our numerical experiments. 

\begin{co}\label{cor:Hedging_demand_Ambiguity}
    Assume that $d=1$ and w.l.o.g.\ that $\theta_1< \ldots <\theta_m$. Further assume that $\E_{\theta_1}[X^*]< \ldots <\E_{\theta_m}[X^*].$ Then $\pi^H$ from Theorem \ref{theo:Ambiguity} is always non-positive, i.e.\ we invest less in the stock under ambiguity.
\end{co}

The proof can be found in the appendix and we will discuss this hedging demand further in the numerical examples section. Another interesting observation is the following.



\begin{co}\label{co:EXH}
The expectation of the second term of \eqref{eq:optwealth_ambiguity}, $X^H$, can be written as
\begin{align*}
    \E[X^H] = \Cov\left(\E_\Theta[\hat\Lambda(T)],\E_\Theta[X^*]\right).
\end{align*}
\end{co}

Note that the structure of $\E[X^H]$ resembles the structure of $\E[X^*]$ on the level of model ambiguity. More precisely we obviously have
\begin{align*}
    \E[X^*] &= \E[X^*] \E[\hat\Lambda(T)] +\E[X^*\hat\Lambda(T)]-\E[X^*\hat\Lambda(T)] \\
    &= x_0- \Cov\left(\hat\Lambda(T),X^*\right).
\end{align*}
Thus, $\E[X^H]$ is exactly the same covariance, but for the respective expected quantities conditioned on $\Theta$.

\subsubsection{Continuous prior}

In this section, we assume that the prior $\Q_0$ is absolutely continuous with respect to the Lebesgue measure on $\R^d$ and denote its density by $q_0$.  Further we denote the (random) density of $\Q_T$ by $q_T(Y,\cdot)$; its explicit formula is given in \eqref{eq:filter}.

Throughout this section, we impose an additional assumption on $q_0$ to ensure the existence of a solution to the optimization problem.

\begin{assumption}\label{assum:q0}
    The prior density $q_0$ satisfies
    \begin{equation*}
        \int_{\R^d} \int_{\R^d}  (\E_x[q_T(Y,y)])^2 \frac{q_0(y)}{q_0(x)} \, \d x \, \d y< \infty.
    \end{equation*}
\end{assumption}

In this context, the optimal terminal wealth and investment strategy depend on the solution $b(x) \in L^2(q_0)$ to the Fredholm integral equation 
\begin{align}
    b(x) + \gamma \int_{\D} b(y) (K_2(x,y) - K_1(y)) \d y = 
    f(x),
    \label{eq:Fredholm}
\end{align}
where $\gamma = \frac{\beta - \alpha}{\alpha}$, $K_1(y) = \E_y[\hat \Lambda(T)] q_0(y)$, $K_2(x,y) = \E_x[q_T(Y,y)]$ and
\begin{equation*}\label{eq:fx}
    f(x) = x_0 - \frac 1{2\alpha} \left(\E[\hat \Lambda^2(T)] - \E_x[\hat\Lambda(T)]\right).
\end{equation*}

Lemma \ref{lem:bx} below discusses the existence of a solution to this equation.

\begin{lem}\label{lem:bx}
    For $\phi \in L^2(q_0)$, let $T_2 \phi = \int_{\mathbb R} \phi(y) K_2(x,y) \d y$ with spectrum $\sigma(T_2)$.\footnote{The resolvent set of an operator $A$ is the set of all its regular values, that is, all $\lambda \in \mathbb C$ such that $(\lambda I - A)^{-1}$ exists and is bounded. The spectrum of $A$ is the complement in $\mathcal C$ of the resolvent set.} If $-\frac{1}{\gamma} \notin \sigma(T_2)$, a solution to \eqref{eq:Fredholm} is given by
    \begin{align}\label{eq:bx}
        b(x) = \psi_0(x) + \gamma \int_{\D} K_1(y) \psi_0(y) \d y,
    \end{align}
    where $\psi_0 = (I + \gamma T_2)^{-1} f$. Furthermore, $T_2$ is self-adjoint on $H$, that is, 
    \begin{align*}
        K(x,y) q_0(x) = K(y,x) q_0(y)
    \end{align*}
    for a.e. $(x,y) \in R^d \times R^d$.
\end{lem}

\begin{rem}
    Assumption \ref{assum:q0} ensures that $T_2$ is a Hilbert-Schmidt operator on $L^2(q_0)$, and therefore compact. The Hilbert-Schmidt property is sufficient but not necessary for compactness of $T_2$. For example, if $\sup_{x \in \R^d} \int_{R^d} \frac{p_{Y(T)|\Theta}(s|x)}{p_{Y(T)}} \, \d s < \infty$, $T_2$ is also compact. Nonetheless, we choose to present the result of this section under Assumption \ref{assum:q0}, which can be shown to hold in the context of the Gaussian prior example presented in Section \ref{ssec:gaussian}.
\end{rem}

\begin{rem}
    Compactness of $T_2$ ensures that its spectrum $\sigma(T_2)$, contains $\lambda = 0$ and $\sigma(T_2)\setminus \{0\}$ is at most countable with no accumulation points; its only elements are the eigenvalues of $T_2$ (see Theorem 3.9 of \cite{kress1989linear}). Thus, \eqref{eq:bx} is a solution to \eqref{eq:Fredholm} if none of the eigenvalues of $T_2$ is equal to $-\frac{1}{\gamma}$.
    Since $K_2$ is a stochastic kernel, we have from the Krein-Rutman theorem (see Theorem 1.2 of \cite{du2006order}) that 
\begin{enumerate}
    \item[(i)] the largest eigenvalue of $T_2$ is simple and equal to $1$, and
    \item[(ii)] all the other eigenvalues $\mu_n$ of $T_2$ satisfy $|\mu_n| < 1$.
\end{enumerate}
It follows that if $\gamma \in (0,1]$, $-\frac{1}{\gamma} \leq -1$, the operator  $(I + \gamma T_2)$ is always invertible, since every eigenvalue of $T_2$ is strictly above $-1$.
On the other hand, if $\gamma > 1$, it is necessary to check numerically that all the eigenvalues of $T_2$ are different from $-\frac{1}{\gamma}$.

For our main theorem of this subsection we denote $G_x(t,y):=F(t,y) \exp(-x^\top y+\frac12 \|x\|^2t)$ for all $x\in\D, y\in\R^d, t\in [0,T]$.
\end{rem}

\begin{theo}\label{theo:Ambiguity_cts}
    Let $\Q_0$ be absolutely continuous with respect to the Lebesgue measure and denote its density by $q_0$. Then, under the conditions of Lemma \ref{lem:bx}, the optimal terminal wealth is given by
    \begin{align} \nonumber
        X^* =&\; x_0-\frac{1}{2\alpha}\Big(\E[\hat\Lambda^2(T)] - \hat\Lambda(T)\Big) \\ \nonumber
    & + \Big(\frac{\alpha-\beta}{\alpha}\Big) \int_{\D} b(x) \left(q_T(Y,x) - \E_x[\hat\Lambda(T)] q_0(x) \right) \d x\\ \label{eq:optwealth_ambiguity_cont}
    = & \; X^B +
     \Big(\frac{\alpha-\beta}{\alpha}\Big) \int_{\D} b(x) \left(q_T(Y,x) - \E_x[\hat\Lambda(T)] q_0(x) \right) \d x,
    \end{align}
    with $b(x)$ defined by \eqref{eq:bx}.
    The optimal investment strategy for the stocks is given by
        \begin{align}\nonumber
    \pi^*(t) =& -(\sigma^\top)^{-1} \frac{1}{2\alpha} \int_{\R^d} \frac{\nabla F(T,Y(t)+z)}{F^2(T,Y(t)+z)} \varphi_{T-t}(z) dz -\\ \nonumber
    &-\Big( \frac{\alpha-\beta}{\alpha}\Big) \int_{\D} b(x) q_0(x) (\sigma^\top)^{-1} \int_{\R^d} \frac{\nabla G_x(T,Y(t)+z)}{G_x^2(T,Y(t)+z)} \varphi_{T-t}(z) dzdx \\ \nonumber
    = & \;\pi^B(t) -\Big( \frac{\alpha-\beta}{\alpha}\Big) \int_{\D} b(x) q_0(x) (\sigma^\top)^{-1} \int_{\R^d} \frac{\nabla G_x(T,Y(t)+z)}{G_x^2(T,Y(t)+z)} \varphi_{T-t}(z) dzdx\\ \label{eq:optinvest_ambiguity_cont}
    = &  \;\pi^B(t)  +\pi^H(t), \quad t\in[0,T]
\end{align}
    
\end{theo}

\begin{rem}
    Numerical methods may be necessary to check that the eigenvalues of $T_2$ do not coincide with $\frac{\alpha}{\alpha-\beta}$, and to invert $(I + \gamma T_2)^{-1}$.
    To compute the eigenvalues of $T_2$, the integral can be discretized using quadrature rules; the eigenvalues of the resulting matrix can be computed numerically. 

\end{rem}

\subsubsection{Example: Gaussian prior}\label{ssec:gaussian}

In this section, we assume that we have one risky asset with a prior distribution of $\Theta$ which is normal $\mathcal{N}(\theta,v^2)$, $\theta\in\R, v>0$ that is $$q_0(z)  = \frac{1}{\sqrt{(2\pi) v^2}} \exp(-\frac{1}{2 v^2} (z-\theta)^2 ), \quad z\in\R.$$ 
To simplify computations, we consider a time horizon $T=1$. Moreover, we set  $\lambda =-\gamma := \frac{\alpha-\beta}{\alpha} < 0$. In this case the function $F: [0,T]\times \mathbb{R}\to\mathbb{R}_+$ is given by (see \cite{jouini_bayesian_2001}, Remark 5.3): 
$$ F(t,y)= (1+tv^2)^{-1/2} \exp\Big( \frac{(\theta +v^2 y)^2}{2v^2 (1+tv^2)}-\frac{\theta^2}{2v^2}\Big).$$
And we  immediately compute
\begin{align*}
l(x) &:= \E_x[\hat\Lambda]= \frac{(1+v^2 )^{}}{(1+2v^2)^{1/2}}  \exp\Big(\frac{2 \theta^2-2\theta x-v^2x^2}{2(1+2v^2)}\Big)\\
q(x,y) &:=K_2(x,y) = \E_x[q_1(Y,y)]=\frac{1+v^2}{\sqrt{2\pi}v(1+2v^2)^{1/2}}  \exp\Big(-\frac{((x-y)v^2+\theta -y)^2}{2v^2(1+2v^2)}\Big) \\
L &:= \E[\hat\Lambda^2] 
= \frac{1+v^2}{\sqrt{2v^2+1}}  \exp\Big(\frac{\theta^2}{1+2v^2}\Big)
\end{align*}
Note that $q$ is indeed a stochastic kernel.
Thus, we can write the Fredholm equation \eqref{eq:Fredholm} in the form
$$ b(x) = x_0-\frac{1}{2\alpha}\Big( L-l(x)\Big) +\lambda \int_{-\infty}^\infty [q(x,y) -\varphi_0(y) l(y)] b(y)dy.$$
Using the probabilistic Hermite polynomials $(H_n)$ which satisfy
\begin{align}\label{eq:Hermite_gen}
\sum_{n=0}^\infty H_n(x) \frac{t^n}{n!}= e^{xt-\frac12 t^2}, \quad t,x\in\R    
\end{align}  we obtain the following solution where we set
\begin{equation}\label{eq:notation0}
\rho=\frac{v^2}{1+v^2}\in(0,1),
  \qquad
  \eta=-\frac{\theta v}{1+v^2}.
\end{equation}
A detailed derivation can be found in the appendix.

\begin{theo}\label{thm:Fredhol_sol_Gaussian}
    The solution of the Fredholm integral equation \eqref{eq:Fredholm} is given by
    $$ b(x) = x_0+\frac{1}{2\alpha} \sum_{n=1}^\infty \frac{c_n}{1-\lambda \rho^n}\Big[ H_n\Big(\frac{x-\theta}{v}\Big)-  H_n\Big(-\frac{\theta}{v}\Big)\Big], \quad x\in\R$$ with
   \begin{equation*}\label{eq:cn_explicit2}
  c_n=
  \sum_{j=0}^{\lfloor n/2\rfloor}
  \frac{(-1)^j\rho^{2j}\eta^{\,n-2j}}
       {2^j j!(n-2j)!},\quad n\in\mathbb{N}.
\end{equation*}    
\end{theo}

\begin{rem} As long as $ |\lambda| < 1/\rho$ we can represent the solution explicitly as a Neumann series, that is
\begin{align*}
    \label{eq:neumann_solution}
  b(x)=x_0+\frac{1}{2\alpha}\sum_{j=0}^{\infty}\lambda^j
       \bigl(F_j(x)-F_j(0)\bigr),
\end{align*}
with
\begin{equation}\label{eq:Fj}
  F_j(x)=\frac{1}{\sqrt{1-r_j^2}}
  \exp\!\left(
  -\frac{r_j^2z^2-2\eta_j z+\eta_j^2}{2(1-r_j^2)}
  \right),
  \qquad z=\frac{x-\theta}{v}
\end{equation}
where
\begin{equation}\label{eq:rj_etaj}
  r_j=\rho^{j+1},
  \qquad
  \eta_j=\eta\rho^j.
\end{equation}
\end{rem}

Having the solution $b$ of the integral equation, we can now use the general main Theorem \ref{theo:Ambiguity_cts} to obtain the optimal terminal wealth. This involves computing two integrals and putting terms together. The optimal investment strategy may also be derived from Theorem \ref{theo:Ambiguity_cts}.
\begin{theo}\label{theo:X_Gaussian}
    In the special Gaussian case of this subsection, the optimal terminal wealth is given by
    \begin{align*}
        X^* &= C_1 + C_2(Y(1)) \hat \Lambda
    \end{align*}
    with
    \begin{align*}
        &C_1 = (1-\lambda) x_0-\frac{L}{2\alpha}- \frac{\lambda}{2\alpha}    \sum_{n=1}^\infty \frac{c_n(\rho^n-1)}{1-\lambda \rho^n} H_n\Big(-\frac{\theta}{v}\Big),\\
        &C_2(Y(1)) = \frac 1{2\alpha} + \frac{\lambda}{\sqrt{v^2+1}}\exp\left(\frac{Y(1)^2}{2} - \frac{(Y(1)-\theta)^2}{2(v^2+1)}\right) \\
    &\qquad\times \left(x_0 + \frac 1{2\alpha} \sum_{n=1}^\infty \frac{c_n}{1-\lambda\rho^n} \left[\left(\frac{v}{\sqrt{1+v^2}}\right)^n H_n\left(\frac{Y(1)-\theta}{\sqrt{1+v^2}}\right) - H_n\left(-\frac{\theta}{v}\right)\right]\right). 
    \end{align*}

\end{theo}

\section{Numerical examples}\label{sec:experiments}

We now present examples illustrating the results of Section \ref{ssec:main_ambiguity}. 
We consider two different priors for $\Theta$: a discrete one with two mass points, which provides a minimal analytically tractable setting, and a Gaussian prior, the natural continuous benchmark.
For each, we examine how the distribution of the final payoff responds to varying degrees of risk and ambiguity aversion. 
In the calibrations considered below, we find that the investor's attitude towards market risk affects the distribution of the final payoff more strongly than her ambiguity aversion. 
Throughout the section, we consider that the market contains a single risky asset ($d=1$). Unless otherwise stated, we set $S(0)=x_0=1$ and $\sigma = 0.2$, with the remaining parameters reported in each subsection.

\subsection{Discrete prior}

In this section, we consider a fixed time horizon $T=0.25$ and a two-point prior; that is $\Q_0(\{0.15\}) = \Q_0(\{0.45\}) = 0.5$, unless otherwise indicated.
This prior can be interpreted as two equally likely economic scenarios: a favorable one ($\theta = 0.45$) and an unfavorable one ($\theta = 0.15$).
The two possible values for $\Theta$ correspond to a drift $\mu \in \{0.03,0.09\}$.
We illustrate the impact of the risk and ambiguity aversion parameters $\alpha$ and $\beta$ on the optimal terminal wealth $X^*_T$ and the associated investment strategy $\pi^*$. The figures of this section were obtained using Lemma \ref{lem:solution_bk} and Theorem \ref{theo:Ambiguity}. 

We also studied a five-point prior, with the same mean but a wider support, as well as a longer time horizon $T=1$. Neither materially changed the results, suggesting that our findings are driven by the prior's mean and dispersion rather than its finer structure or the specific horizon. We therefore report only the two-point, $T=0.25$ case.

\subsubsection{Optimal terminal wealth $X^*$}

Figure \ref{fig:Xdist_discrete_prior} presents the distribution of the terminal wealth $X^*$ for various values of $\alpha$ and $\beta$, with the case $\alpha = \beta = -1$ representing the ambiguity neutral setting of Section \ref{ssec:bayesian} with parameter $\alpha = -1$. We recall that $\alpha$ and $\beta$ must satisfy $\beta < \alpha < 0$, and that $\alpha$ represents the level of risk aversion while $\beta$ or $\alpha-\beta$ is linked to ambiguity preferences. The mean and the variance of each distributions are reported in Table \ref{tab:summary_stats}.

\begin{figure}[H]
	\includegraphics[scale=0.8]{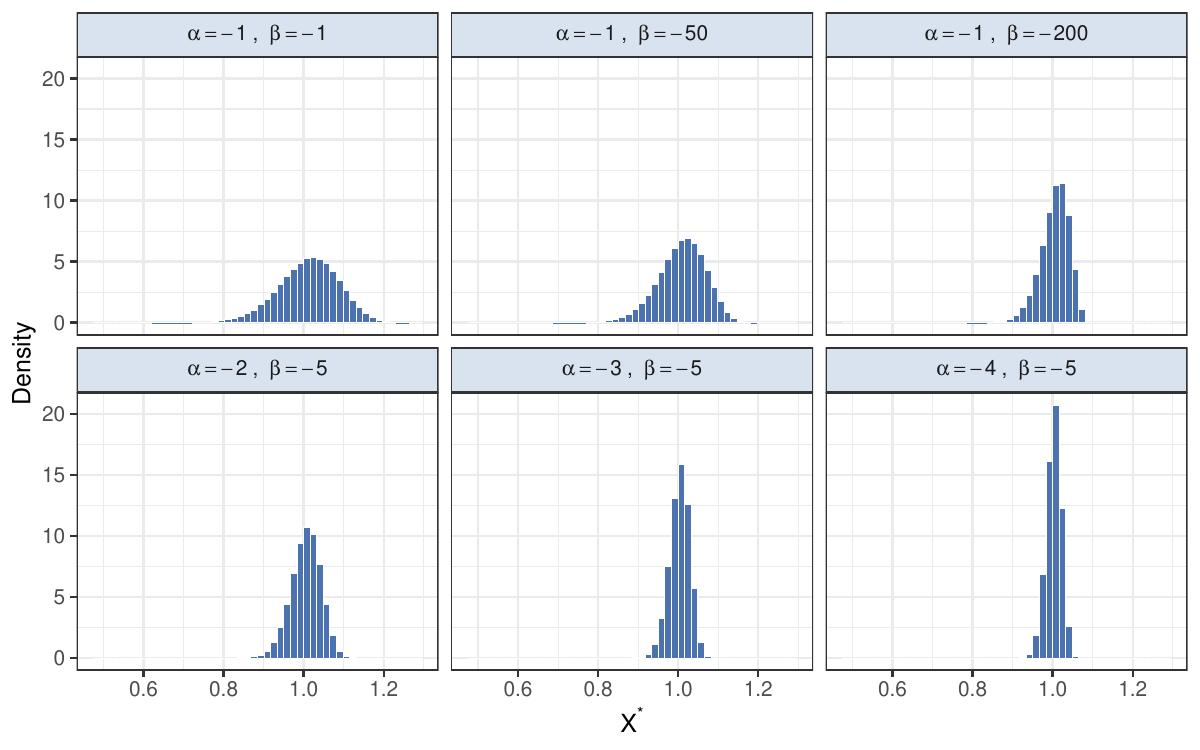}
	\caption{Density of the terminal wealth $X^*$ for different values of $\alpha$ and $\beta$, discrete prior, $T=0.25$. The case $\alpha = \beta$ represents the ambiguity neutral setting of Section \ref{ssec:bayesian} with parameter $\alpha$.}\label{fig:Xdist_discrete_prior}
\end{figure}

\begin{table}[H]
\centering
\caption{Empirical mean and variance of $X^*$ for various $(\alpha, \beta)$ couples.}
\label{tab:summary_stats}
\begin{tabular}{cccc}
\toprule
$\alpha$ & $\beta$ & Mean & Variance \\
\midrule
$-1$ & $-1$   & 1.01 & $5.63 \times 10^{-3}$  \\
$-1$ & $-50$  & 1.01 & $3.48 \times 10^{-3}$\\
$-1$ & $-200$ & 1.01 & $1.29 \times 10^{-3}$\\
$-2$ & $-5$   & 1.01 & $1.39 \times 10^{-3}$\\
$-3$ & $-5$   & 1.00 & $6.22 \times 10^{-4}$\\
$-4$ & $-5$   & 1.00 & $3.51 \times 10^{-4}$\\
\bottomrule
\end{tabular}
\end{table}

A visual inspection of Figure \ref{fig:Xdist_discrete_prior} allows use to conclude that in all cases, the distribution of $X^*$ has a slight left skew; the left tail appears to be heavier. Table \ref{tab:summary_stats} shows that both the empirical mean and variance decrease as both risk and ambiguity aversion increase (that is, as the absolute value of $\alpha$ and $\alpha-\beta$ increase). This was expected; higher aversion towards market risk or model ambiguity result in a preference for a ``less random'' payoff, at the expense of higher returns.

A variation in $\alpha$ has a much stronger impact on the distribution of $X^*$ than one in $\beta$. For example, the distributions obtained using $(\alpha,\beta) = (-1,-200)$ and $(\alpha,\beta) = (-2,-5)$ are very similar. Thus, in this problem, the optimal payoff for an investor with a very high level of ambiguity aversion ($\beta = -200$) is also (almost) optimal for an investor with a slightly increased level of risk and ambiguity aversion ($\alpha = -2$, $\beta = -5$).

Figure \ref{fig:XST_discrete_prior} plots $X^*$ as a function of $S(T)$, for different levels of risk and ambiguity aversion parameters $\alpha$ and $\beta$. Unsurprisingly, lower levels of aversion to either type of risk leads to more variation in $X^*$ as the value of the risky asset changes. For most parameters studied, $X^*$ is an increasing function of $S(T)$. However, for very high levels of ambiguity aversion ($\beta = -200$), $X^*$ is non-monotone and decreases for very large values of $S(T)$. The reluctance of very ambiguity averse investors to invest in the risky asset may explain this payoff (see also Figure \ref{figAmbfromSt}). 

\begin{figure}[H]
	\includegraphics[scale=0.8]{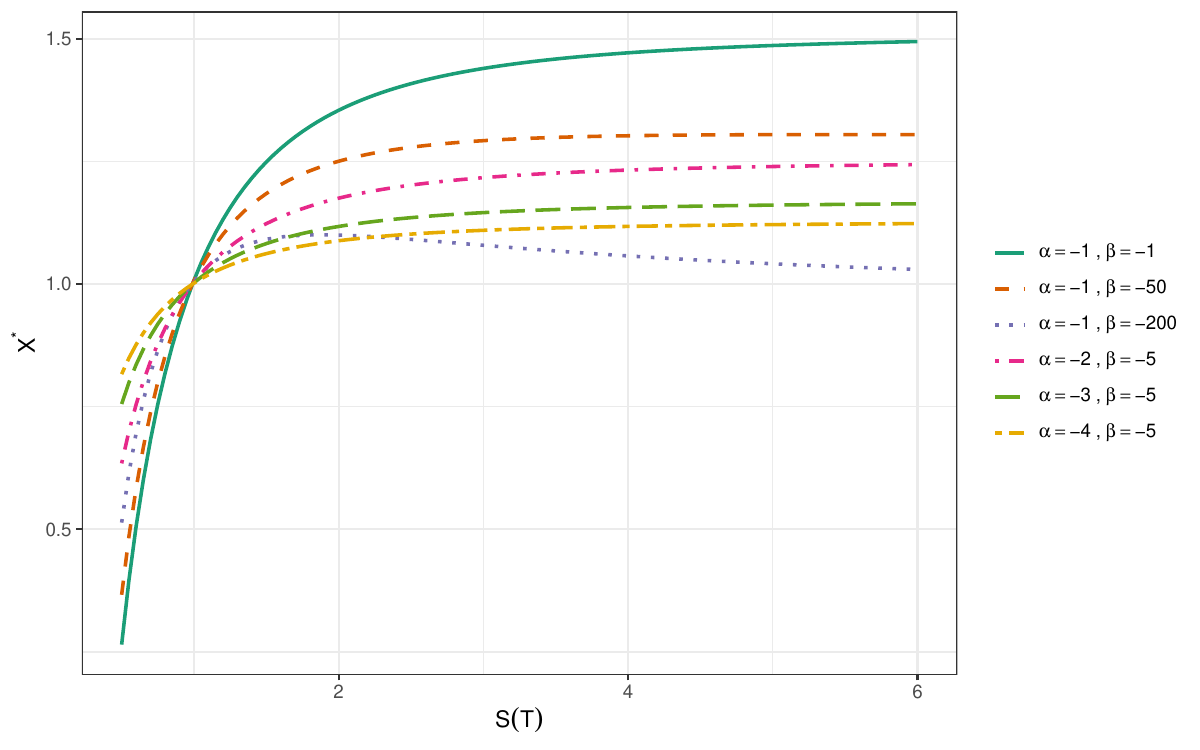}
	\caption{Optimal final wealth $X^*$ as a function of $S(T)$, $T=0.25$, discrete prior.}\label{fig:XST_discrete_prior}
\end{figure}

\subsubsection{Investment strategy}

Figure \ref{figpaths} shows two sample paths of the risky asset and the resulting optimal investment strategies $\pi^B$ and $\pi^*$, representing the optimal strategies of ambiguity neutral and ambiguity averse investors, respectively. In both examples, the value of the risky asset $S$ stays in the interval $(0.9,1.15)$. For these values, the difference between both strategies stays relatively constant, with the ambiguity averse investment strategy being about 0.15 points below the ambiguity neutral one. This lower exposure to the risky asset explains the lower variance observed when comparing the second to the first row of Table \ref{tab:summary_stats}.

\begin{figure}[H]
	\centering
	\begin{tabular}{cc}
		\includegraphics[scale=0.3]{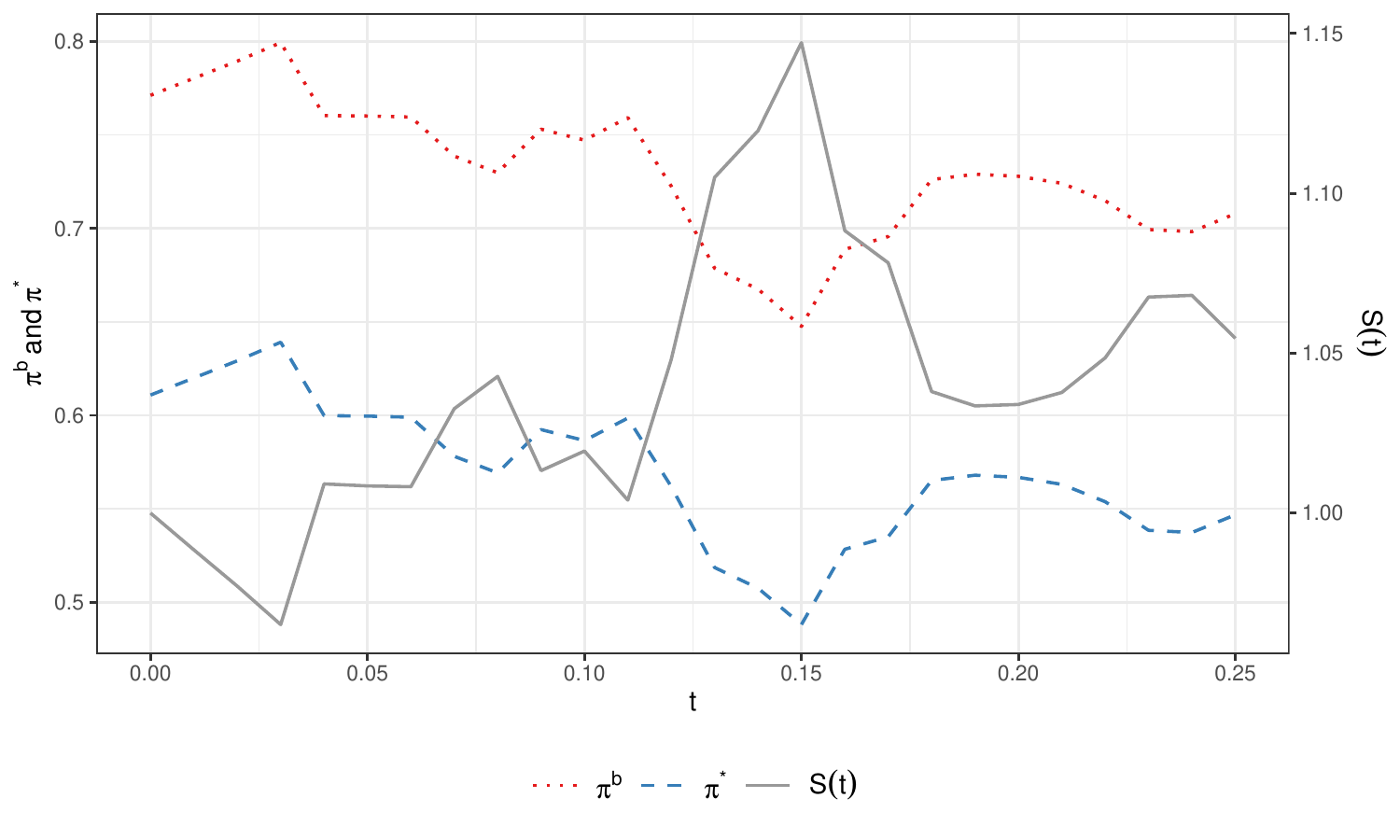}&
		\includegraphics[scale=0.3]{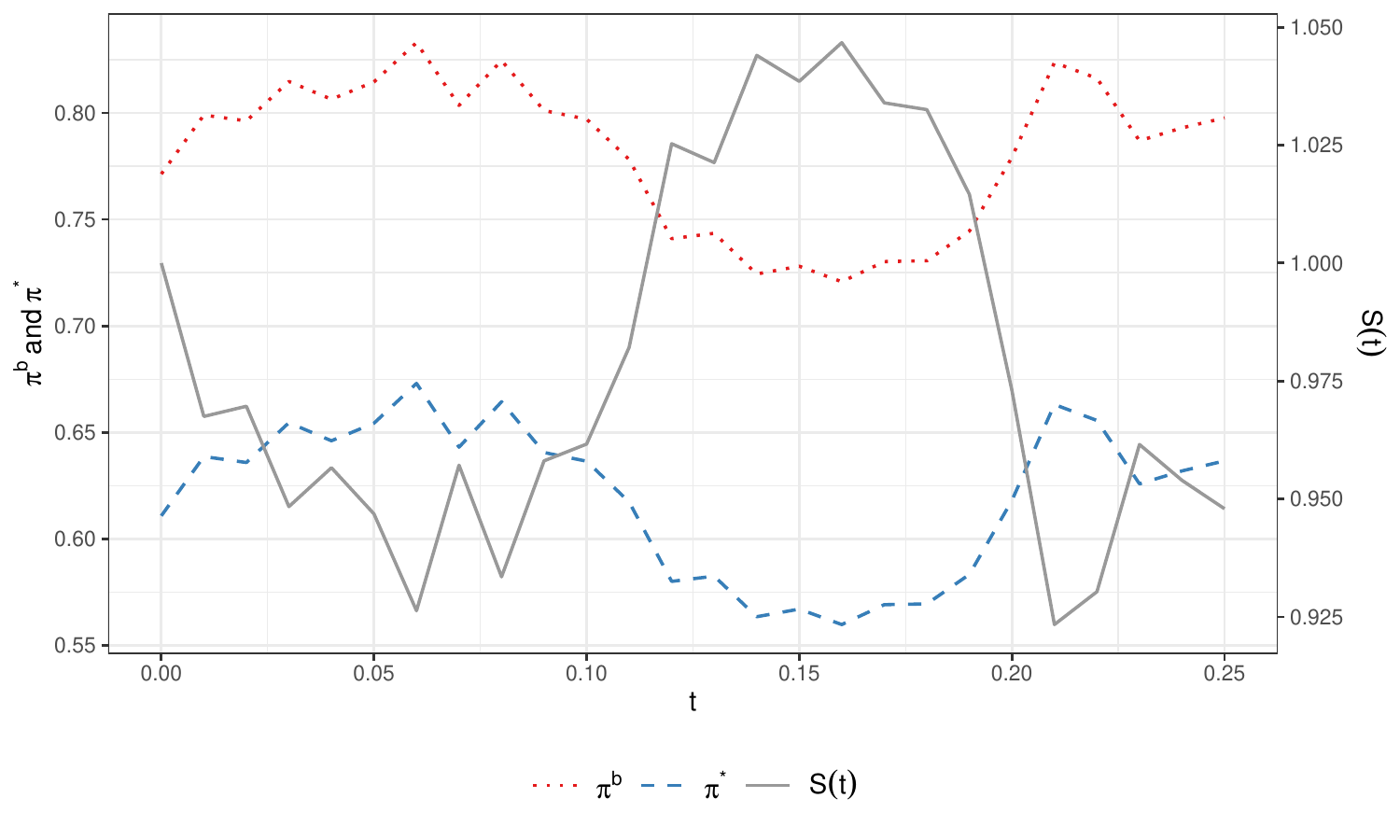}\\
	\end{tabular}
	\caption{Sample paths of the risky asset and the associated optimal investment strategies $\pi^B$ and $\pi^*$, $\alpha = -1$ and $\beta = -50$.}\label{figpaths}
\end{figure}

Figures \ref{figAmbfromSt} and \ref{figAmbfromSt_small} show $\pi^H(t)$, the difference between the proportion invested in the risky asset by the ambiguity averse and the ambiguity neutral investor, at $t=0.1$ and for different values $\alpha$ and $\beta$. Figure \ref{figAmbfromSt} uses the same risk and ambiguity aversion levels as the previous figures, while Figure \ref{figAmbfromSt_small} uses smaller (absolute) values for $\beta$, to better illustrate the smaller values of $\pi^H(t)$. 

\begin{figure}[h!]
	\centering
		\includegraphics[scale=0.75]{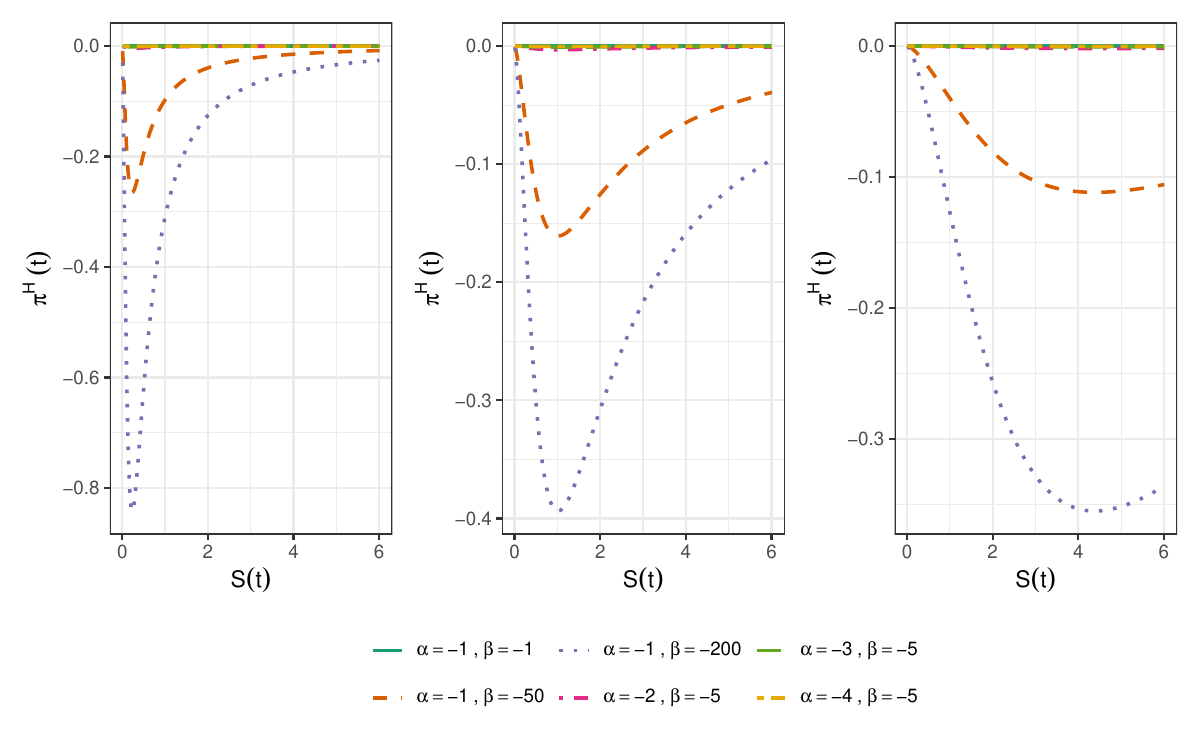}
	\caption{Difference $\pi^H(t) = \pi^*(t)-\pi^B(t)$, with $t=0.1$. Left: $\Q_0(0.15)=0.1$. Middle: $\Q_0(0.15)=0.5$. Right: $\Q_0(0.15)=0.9$.}\label{figAmbfromSt}
\end{figure}

\begin{figure}[h!]
	\centering
		\includegraphics[scale=0.75]{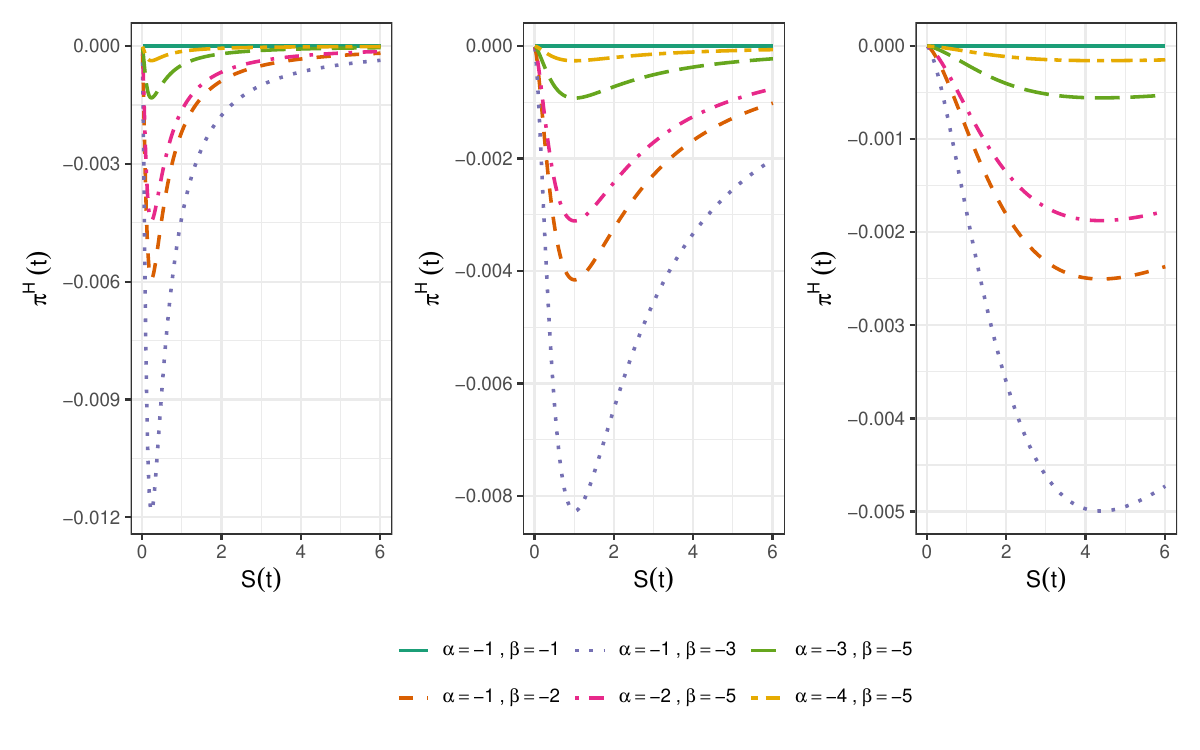}
	\caption{Difference $\pi^H(t) = \pi^*(t)-\pi^B(t)$, with $t=0.1$. Left: $\Q_0(0.15)=0.1$. Middle: $\Q_0(0.15)=0.5$. Right: $\Q_0(0.15)=0.9$.}\label{figAmbfromSt_small}
\end{figure}

In both figures, $\pi^H(t)$ is more sensitive to $\alpha$ than to $\beta$, which may be counterintuitive when comparing these results to Figure \ref{fig:Xdist_discrete_prior}. 
However, it is important to note that Figures \ref{figAmbfromSt} and \ref{figAmbfromSt_small} isolate the ambiguity hedging part of the strategy, while the distribution in Figure \ref{fig:Xdist_discrete_prior} concerns the full payoff, which relies heavily on the Bayesian payoff $X_T^B$.
Thus, the sensitivity of the distribution of $X^*_T$ to $\beta$ comes from the ``ambiguity neutral'' part of the payoff, $X^B_T$, which dominates it.
Considering only $\pi^H$, allows to demonstrate the impact of $\alpha-\beta$ on the part of the hedging strategy that is specifically concerned with ambiguity. 
Nonetheless, for moderate levels of ambiguity aversion (Figure \ref{figAmbfromSt_small}), the impact of $\alpha$ is still modest, with the largest difference reaching only 1.2\% when $\beta = -3$.
Interestingly, for the same level of ambiguity aversion, larger (absolute) values of $\alpha$ reduce the difference between the ambiguity averse and the ambiguity neutral investment strategy.
This phenomenon can be traced back to the factor $\frac{\alpha-\beta}{\alpha}$ in $\pi^H$.

For this analysis, we also considered two other priors, with the same support but difference probabilities. The left-most plots of Figures \ref{figAmbfromSt} and \ref{figAmbfromSt_small} use $\Q_0(0.15) = 0.1 = 1-\Q_0(0.45)$, while the right-most ones use $\Q_0(0.15) = 0.9 = 1-\Q_0(0.45)$. That is, on the left, the scenario with a higher drift is more likely, while on the right, the lower drift one has a higher probability. 
The prior distribution has a significant impact on the investment strategy; the ambiguity hedge $\pi^H(t)$ is larger when the market is not behaving as predicted by the most likely scenario, as described by the prior distribution.

\subsection{Gaussian prior}

The analysis in the present subsection assumes a Gaussian prior (see Section \ref{ssec:gaussian}) with parameters $\theta = 0.3$ and $v = 0.15$, so that the expectation and the variance of the distribution match those of the discrete prior studied in the previous section.
We also consider a time horizon $T=1$, as in Section \ref{ssec:gaussian}.

Figure \ref{fig:Xdist_gaussian} presents the distribution of the optimal payoff at $T$ for different values of $\alpha$ and $\beta$. Similarly to the case of the discrete prior, the distribution is significantly more sensitive to changes in $\alpha$ than in $\beta$. Again, this is due to the optimal terminal wealth being dominated by the Bayesian optimal payoff $X^B$, which is scaled by $\frac{1}{2\alpha}$, see Theorem \ref{theo:Ambiguity_cts}.
In all cases, the distributions are skewed negatively, more so than in the discrete prior case. For our parametrization, higher (absolute) values of $\beta$ appear to lead to an upper bounded distribution (see also Figure \ref{fig:XstarGaussian}).

Our numerical examples suggest that in the Gaussian prior case, the distribution of $X^*$ is slightly more sensitive to changes in $\beta$ compared to the case of the discrete prior. 
This might be explained by the wider support of the Gaussian prior (the whole real line, compared to two isolated points for the discrete prior studied), which makes it harder for the investor to ``learn'' the true value of $\Theta$ as the market develops.
Therefore, ambiguity remains and has a bigger impact on the optimal terminal wealth.

\begin{figure}[h!]
	\centering
		\includegraphics[scale=0.78]{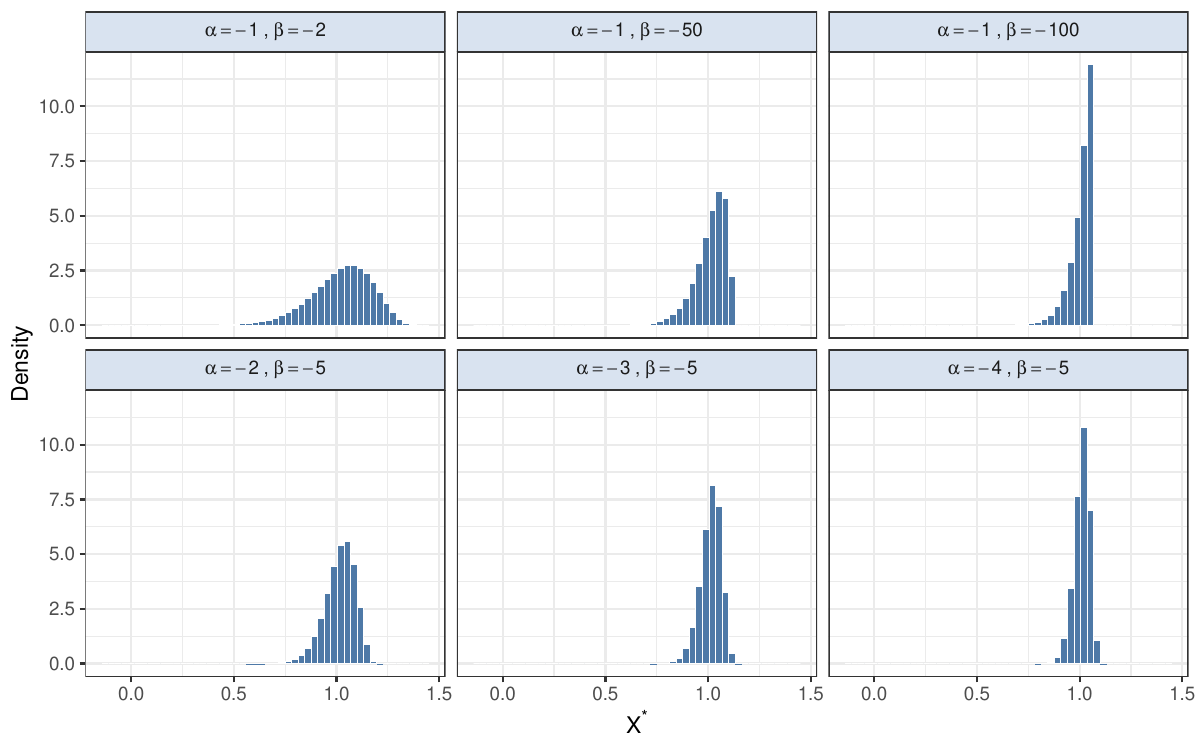}
	\caption{Distribution of the terminal wealth $X^*$ for different values of $\alpha$ and $\beta$, Gaussian prior, $T=1$. The case $\alpha = \beta$ represents the ambiguity neutral setting of Section \ref{ssec:bayesian} with parameter $\alpha$.}\label{fig:Xdist_gaussian}
\end{figure}

Figure \ref{fig:XstarGaussian} expresses $X^*$ as a function of the terminal value of the risk asset $S(T)$. 
As expected, a higher level of risk and/or ambiguity aversion leads to a smaller range for the payoff, which is also observed in the empirical distribution of Figure \ref{fig:Xdist_gaussian}.
As in the discrete prior case, non-monotonicity of $X^*$ as a function of $S(T)$ is observed for higher levels of risk aversion. 
In the Gaussian prior case, with our parametrization, this phenomenon occurs for lower (absolute) values of $\beta$. 
It explains the boundedness of the empirical distributions of Figure \ref{fig:Xdist_gaussian}.
Such a bounded payoff significantly reduces the variance of $X^*$ and contributes to maximizing the objective function \eqref{eq:CX2}.
Although the optimal investment strategy has not been explicitly derived for the Gaussian case, we can infer from Figure \ref{fig:XstarGaussian} that a strongly ambiguity averse investor avoids a large long position in the risky asset when it is performing well.

\begin{figure}[h!]
	\centering
		\includegraphics[scale=0.78]{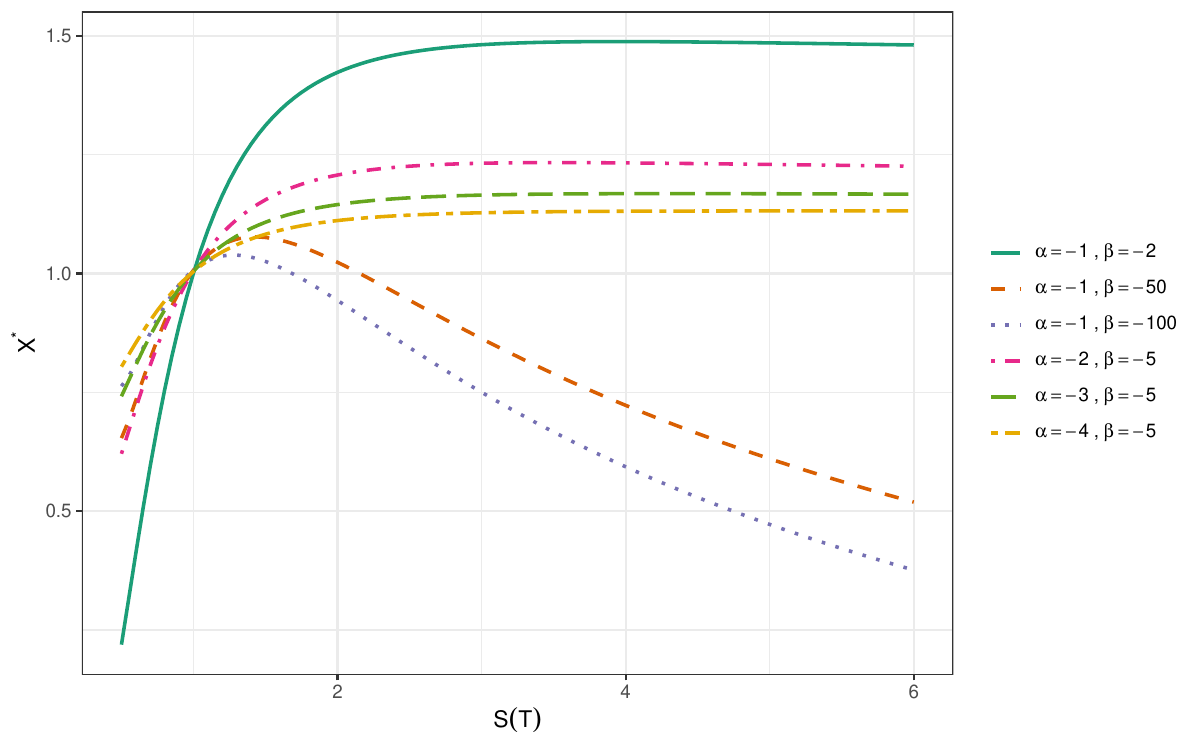}
	\caption{Optimal final wealth $X^*$ as a function of $S(T)$, $T=1$, Gaussian prior.}\label{fig:XstarGaussian}
\end{figure}

\section{Conclusions}
\label{sec:conclusions}

We considered a portfolio allocation problem under model ambiguity using the mean-variance criterion of \cite{maccheroni_alpha_2013}, which separates market risk and model ambiguity aversion. We derived the optimal portfolio when the prior distribution of the (unknown) drift term is discrete and continuous. In the discrete prior case, the optimal terminal wealth relies on the solution of a system of linear equations; the equivalent in the continuous prior case is a Fredholm integral equation. In both cases, under mild conditions on the parameters, the optimal terminal wealth exists. From this result, we also derived the associated investment strategy. In the continuous case, we gave a specific example by assuming a Gaussian prior distribution.

Numerical results show that the final payoff is significantly more sensitive to market risk aversion. However, the part of the investment strategy that hedges against model ambiguity depends not only on the level of ambiguity aversion, but also on the initial beliefs on the distribution of the drift term (its prior distribution).

\section{Declaration}
The authors used ChatGPT Pro in April/May 2026 to solve the Fredholm integral equation in the special case of a  normal prior and to compute the integrals to obtain $X^*$, as well as Claude Pro to assist with the production of numerical results. All resulting calculations have been checked by hand. The authors assume responsibility for all content.

\section{Appendix}
\subsection{Proof of Lemma \ref{lem:interchange}}
To ease notation instead of $X^\pi(T)$ we simply write $X^\pi$. The inequality
   \begin{align*}
& \sup_{\pi \in \Pi(x_0)} \inf_{b \in L^2(\D)}  \E[X^\pi]+\beta  \Var(X^\pi) +(\alpha-\beta) \E[(X^\pi-b(\Theta))^2]\\
\le  &   \inf_{b \in L^2(\D)}  \sup_{\pi \in \Pi(x_0)} \E[X^\pi]+\beta  \Var(X^\pi)+(\alpha-\beta) \E[(X^\pi-b(\Theta))^2]
   \end{align*} 
is obvious. Now we show '$\ge$'.  The value of the optimization problem is finite.
Thus, for arbitrary $\varepsilon>0$ there exists an admissible $X^*$ s.t.
\begin{align*}
    & \E[X^*]+\beta  \Var(X^*)+(\alpha-\beta) \E[\Var(X^*|\Theta)] \\
   \ge \;&  \sup_\pi \E[X^\pi]+\beta  \Var(X^\pi)+(\alpha-\beta) \E[\Var(X^\pi|\Theta)] -\varepsilon
\end{align*}
For the first expression we obtain
\begin{align*}
    & \E[X^*]+\beta  \Var(X^*)+(\alpha-\beta) \E[\Var(X^*|\Theta)] \\
    =\; &  \inf_b \E[X^*]+\beta  \Var(X^*)+(\alpha-\beta) \E[(X^*-b(\Theta))^2] \\
    \le \;& \sup_\pi \inf _b \E[X^\pi]+\beta  \Var(X^\pi)+(\alpha-\beta) \E[(X^\pi-b(\Theta))^2]
\end{align*}
For the second expression we obtain with $b^\pi(\Theta) = \E[X^\pi|\Theta]$
\begin{align*}
    & \sup_\pi \E[X^\pi]+\beta  \Var(X^\pi)+(\alpha-\beta) \E[\Var(X^\pi|\Theta)]-\varepsilon\\
   =\;  & \sup_\pi \E[X^\pi]+\beta  \Var(X^\pi)+(\alpha-\beta) \E[(X^\pi-b^\pi(\Theta))^2]-\varepsilon\\
    \ge \;  & \inf_b \sup_\pi \E[X^\pi]+\beta  \Var(X^\pi)+(\alpha-\beta) \E[(X^\pi-b(\Theta))^2]-\varepsilon
\end{align*}
Letting $\varepsilon\downarrow 0$ implies the result.

\subsection{Proof of Lemma \ref{lem:solution_bk}}

To verify the conditions under which $(I-\frac{\alpha-\beta}{\alpha} A)$ is invertible, we re-write the expression as 
\begin{equation}
    \tilde A_1 - \frac{\alpha-\beta}{\alpha} \mathbf 1_m \varsigma^\top,
\end{equation}
with $\tilde A_1 = I - \frac{\alpha-\beta}{\alpha} A_1$ and we observe that if $\tilde A_1$ is invertible, we can use the so-called matrix determinant lemma (see Lemma 1.1 of \cite{ding2007eigenvalues}) to obtain
\begin{equation*}
    \det\left(I-\frac{\alpha-\beta}{\alpha} A\right) = 
    \det\left( \tilde A_1\right) \left(1+ \frac{\alpha-\beta}{\alpha} \varsigma^\top \tilde A_1^{-1} \mathbf 1_m\right).
\end{equation*}
Thus, when $\tilde A_1$ is invertible, $\det\left(I-\frac{\alpha-\beta}{\alpha} A\right) \neq 0$ if and only if \eqref{eq:condition_bk}.
To conclude, we observe that if $\lambda$ is an eigenvalue of $A_1$, then $\left(1-\frac{\alpha-\beta}{\alpha} \lambda\right)$ is an eigenvalue of $\tilde A_1$ and is non-zero as long as $\lambda \neq \frac{\alpha}{\alpha-\beta}$.

\subsection{Proof of Theorem \ref{theo:Ambiguity}}
We use the representation of the objective function derived in \eqref{eq:MVobjective:Afinal}:
\begin{align*}
   & \E[X^\pi(T)]+\alpha  \E[\Var(X^\pi(T)|\Theta)]+\beta \Var(\E[X^\pi(T)|\Theta]) \\ =
 & \E[X^\pi(T)]+\beta  \Var(X^\pi(T))+(\alpha-\beta) \E[\Var(X^\pi(T)|\Theta)] 
   \\ = & \E[X^\pi(T)]+\beta \inf_{\bar b} \E[(X^\pi(T)-\bar b)^2 ]+(\alpha-\beta) \sum_{k=1}^d  \inf_{b_k}\E_{\theta_k} [(X^\pi(T)-b_k)^2] \Q_0(\theta_k)
\end{align*} 
where $\bar b= \E[X^\pi(T)]$ and $b_k = \E_{\theta_k} [(X^\pi(T)].$
In particular note that $\alpha-\beta>0.$ 
The objective function can further be written as 
   \begin{align*}
    & \E[X^\pi(T)]+\beta \E[(X^\pi(T)-\bar b)^2 ]+(\alpha-\beta) \sum_{k=1}^d  \E \big[ (X^\pi(T)-b_k)^2 1_{[\Theta=\theta_k]}\big]\\
   = & \E[X^\pi(T)]+\beta \E[(X^\pi(T)-\bar b)^2 ]+(\alpha-\beta) \sum_{k=1}^d  \E\Big[\E \big[ (X^\pi(T)-b_k)^2 1_{[\Theta=\theta_k]} | \F_T^Y\big]\Big]\\
   = & \E[X^\pi(T)]+\beta \E[(X^\pi(T)-\bar b)^2 ]+(\alpha-\beta) \sum_{k=1}^d  \E \big[ (X^\pi(T)-b_k)^2 \Q_T(\theta_k) \big].
\end{align*}
For the last equation note that $X^\pi(T)$ is $\F_T^Y$-measurable. Using Lemma \ref{lem:interchange} we first maximize this expression w.r.t.\ $X^\pi(T)$, respecting the budget constraint. We use the dual or martingale approach. Once we have the optimal terminal wealth we can get the optimal strategy by replicating it, since the market is complete. Note that we have to make sure that in the end our solution satisfies  $\E[X^\pi(T)]=\bar b$ and $\E_{\theta_k}[X^\pi(T)]=b_k.$ I.e. among all solutions that we get when maximizing this expression under the constraint $\E[\hat \Lambda(T) X^\pi(T)] = x_0$ only those are feasible which satisfy the additional constraints.  Thus, the Lagrange function (with multiplier $\lambda\in\R$) has the form (from now on we skip the time index $T$ and the strategy $\pi$ in the notation of $X^\pi(T)$ and $\hat\Lambda(T)$)
\begin{align*}
    L(X,\lambda) &= \E\Big[X+\beta(X-\bar b)^2 +(\alpha-\beta)\sum_{k=1}^d (X-b_k)^2 \Q_T(\theta_k)+\lambda(\hat\Lambda X-x_0)  \Big].
\end{align*}
Taking  the Fr\'echet derivative w.r.t. $X$ and setting it to zero yields
\begin{align*}
    1+2\beta(X-\bar b) +2(\alpha-\beta)\sum_{k=1}^d (X-b_k) \Q_T(\theta_k)+\lambda\hat\Lambda =0. 
\end{align*}
Since $\sum_{k=1}^d  \Q_T(\theta_k)=1$ we obtain
\begin{align*}
    X= \frac{1}{2\alpha} \Big( 2\beta \bar b-1 +2 (\alpha-\beta)\sum_{k=1}^d b_k \Q_T(\theta_k)-\lambda\hat\Lambda \Big).
\end{align*}
The Lagrange multiplier $\lambda$ and the parameters $\bar b,b_k$ can be obtained from the conditions 
\begin{itemize}
    \item[(i)] $\E[\hat\Lambda X]=x_0.$
     \item[(ii)] $\E[ X]=\bar b.$
      \item[(iii)] $\E_{\theta_k}[ X]=b_k, \; \mbox{ for }\; k=1,\ldots,m.$
\end{itemize}
Consider now the second condition. Note that $\E[\Q_T(\theta_k)]=\Q_0(\theta_k)$, since $(\Q_t(\theta_k))$ is an $(\P,\F^Y)$ martingale and that $\sum_{k=1}^d b_k \Q_0(\theta_k)=\bar b.$ Then (ii) again implies that $\lambda =-1.$
Next we obtain that
$$ \E[\hat \Lambda \Q_T(\theta_k)]= \E\Big[\hat \Lambda \E[1_{[\Theta=\theta_k]}| \F_T^Y]\Big] = \E[\hat\Lambda 1_{[\Theta=\theta_k]}]=\Q_0(\theta_k) \E_{\theta_k}[\hat\Lambda]$$
With this observation condition (i) yields
\begin{align*}
    \bar b = \frac{1}{2\beta} \Big(2\alpha x_0 - \Var(\hat\Lambda) -2(\alpha-\beta) \sum_{k=1}^d \Q_0(\theta_k) b_k \E_{\theta_k}[\hat\Lambda]  \Big)
\end{align*}
Plugging $\lambda$ and $\bar b$ into the equation for $X$ yields
\begin{align}\label{eq:optimalwealth_a}
    X= x_0 - \frac{1}{2\alpha} \Big(\E[\hat\Lambda^2]-\hat\Lambda  - 2 (\alpha-\beta)\sum_{k=1}^d b_k (\Q_T(\theta_k)-\Q_0(\theta_k)\E_{\theta_k}[\hat\Lambda]) \Big).
\end{align} 
Finally, (iii) leads to the linear equation for $b=(b_k):$

$$\Big(I-\frac{\alpha-\beta}{\alpha} A\Big) b = c $$
where $A\in\R^{m\times m}$ has elements
$$ a_{ik} = 
\Big(\E_{\theta_i}[\Q_T(\theta_k)] - \Q_0(\theta_k)\E_{\theta_k}[\hat\Lambda]\Big)$$
and $c\in\R^m$ has elements
$$ c_k = x_0-\frac{1}{2\alpha} (\E[\hat\Lambda^2] - \E_{\theta_k}[\hat\Lambda]) . $$
From Lemma \ref{lem:solution_bk} we know that under our assumptions there is a solution $b$. Plugging in the expressions for $\hat\Lambda$ and $\Q_T(\theta_k)$ from Section \ref{sec:financial_market} into \eqref{eq:optimalwealth_a} yields the representation for $X$.

In order to determine the optimal investment strategy we have to derive the optimal wealth process under $\tilde \P.$ We can then determine the strategy $\pi^*$ from identifying
\begin{align*}
    \hat X_t = \tilde\E[X|\F_t^Y] = x_0+ \int_0^t \pi_s^* \sigma dY_s.
\end{align*}
Using the expression \eqref{eq:optimalwealth_a} we obtain
\begin{align*}
   \tilde\E[X|\F_t^Y] &=  x_0 - \frac{1}{2\alpha} \Big(\E[\hat\Lambda^2]-\tilde\E[\hat\Lambda|\F_t^Y]  - 2 (\alpha-\beta)\sum_{k=1}^d b_k (\tilde\E[\Q_T(\theta_k)|\F_t^Y]-\Q_0(\theta_k)\E_{\theta_k}[\hat\Lambda]) \Big).
\end{align*}
And we have
\begin{align*}
    \tilde\E[\hat\Lambda|\F_t^Y] &= \tilde\E[F^{-1}(T,Y(T))|\F_t^Y] = \int_{\R^d} \frac{1}{F(T,Y(t)+z)} \varphi_{T-t}(z) dz\\
    \tilde\E[\Q_T(\theta_k)|\F_t^Y] &= \Q_0(\theta_k)  \int_{\R^d} \frac{L_k(T,Y(t)+z)}{F(T,Y(t)+z)} \varphi_{T-t}(z) dz
\end{align*}
Inserting these expressions above and using It\^{o}'s lemma to derive the $\ldots dY_s$ term we end up with the stated strategy. Note that this strategy also satisfies $ \int_0^T \E\|\pi^*(t)\|^2 dt< \infty$ due to the properties of the normal distribution.

\subsection{Proof of Corollary \ref{cor:Hedging_demand_Ambiguity}}
We assume here that $d=1,$ i.e.\ there is only one stock. Recall from Theorem \ref{theo:Ambiguity} that in this case the hedging demand is given by
\begin{align}\label{eq:piH}
  \pi^H(t)=  -\frac1\sigma \Big( \frac{\alpha-\beta}{\alpha}\Big) \sum_{k=1}^m b_k \Q_0(\theta_k)  \int_{\R} \frac{ G'_k(T,Y(t)+z)}{G_k^2(T,Y(t)+z)} \varphi_{T-t}(z) dz
\end{align}
Since $\beta<\alpha$ we show that the sum is non-positive. In order to ease notation we set $Y(t)=:y$. Computing the derivative of $G_k$ we obtain
\begin{align*}
 \frac{ G'_k(T,y+z)}{G_k^2(T,y+z)} &=    \frac{\sum_{j=1}^m \Q_0(\theta_j)  \theta_j \exp\big(\theta_j(z+y)-\frac12 \theta_j^2 T\big)-\theta_k F(T,y+z)}{F^2(T,y+z)} \exp\big(\theta_k(z+y)-\frac12 \theta_k^2 T\big)
\end{align*}
Next, note that
\begin{align*}
   & \exp\big(\theta_k (z+y)-\frac12 \theta_k^2T\big) \varphi_{T-t}(z) = \exp\big(\theta_k y  -\frac12 \theta_k^2t\big) \varphi_{\theta_k(T-t),T-t}(z),
\end{align*}
where $\varphi_{\mu,\sigma^2}$ is the normal density with mean $\mu$ and variance $\sigma^2.$
Thus, the integral part in $\pi^H(t)$ can be written as
\begin{align*}
 &  \int_{\R} \frac{ G'_k(T,Y(t)+z)}{G_k^2(T,Y(t)+z)} \varphi_{T-t}(z) dz =  \\
 =& \exp\big(\theta_k y  -\frac12 \theta_k^2t\big)  \int_{\R} \frac{\sum_{j=1}^m \Q_0(\theta_j) \exp\big(\theta_j(z+y)-\frac12 \theta_j^2 T\big)(\theta_j-\theta_k)}{F^2(T,y+z)}\varphi_{\theta_k(T-t),T-t}(z) dz
\end{align*}
Plugging this expression into the sum in \eqref{eq:piH} and rearranging terms yields
\begin{align}\label{eq:piH2}
   & \sum_{k=1}^m \sum_{j=1}^m \Q_0(\theta_k)\Q_0(\theta_j) (\theta_j-\theta_k) b_k\exp\big(\theta_k y  -\frac12 \theta_k^2t\big) \int \frac{ \exp\big(\theta_j(z+y)-\frac12 \theta_j^2 T\big)}{F^2(T,z+y)}  \varphi_{\theta_k(T-t),T-t}(z)dz
\end{align}
Next, note that
\begin{align*}
  & \exp\big(\theta_j(z+y)-\frac12 \theta_j^2 T\big) \varphi_{\theta_k(T-t),T-t}(z)
   = \exp\big(\theta_jy-\frac12 \theta_j^2 t+\theta_j\theta_k (T-t)\big) \varphi_{(\theta_j+\theta_k)(T-t),T-t}(z)
\end{align*}
Using this in \eqref{eq:piH2} yields
\begin{align*}
   & \sum_{k=1}^m \sum_{j=1}^m \Q_0(\theta_k)\Q_0(\theta_j) (\theta_j-\theta_k) b_k\exp\big((\theta_k+\theta_j) y  -\frac12 (\theta_k^2+\theta_j^2)t+\theta_j\theta_k(T-t)\big) \\
   & \hspace{2cm} \cdot \int \frac{\varphi_{(\theta_j+\theta_k)(T-t),T-t}(z) }{F^2(T,z+y)} dz
\end{align*}
Let us define for $j,k=1,\ldots,m$
\begin{align*}
    a_{kj} := \Q_0(\theta_k)\Q_0(\theta_j) \exp\big((\theta_k+\theta_j) y  -\frac12 (\theta_k^2+\theta_j^2)t+\theta_j\theta_k(T-t)\big) \int \frac{\varphi_{(\theta_j+\theta_k)(T-t),T-t}(z) }{F^2(T,z+y)} dz
\end{align*}
Obviously $a_{j,k} >0$ for all $k,j$ and $a_{j,k}=a_{k,j}.$ Thus, we can continue working on our sum
\begin{align*}
 &   \sum_{k=1}^m \sum_{j=1}^m a_{k,j}(\theta_j-\theta_k) b_k =  \sum_{k=2}^m \sum_{j=1}^{k-1} a_{k,j}(\theta_j-\theta_k) b_k +  \sum_{k=1}^{m-1} \sum_{j=k+1}^m a_{k,j}(\theta_j-\theta_k) b_k\\
 =&  \sum_{k=2}^m \sum_{j=1}^{k-1} a_{k,j} (\theta_j-\theta_k) (b_k-b_j).
\end{align*}
With the assumed ordering of the $b_k's$ the statement follows.

\subsection{Proof of Corollary \ref{co:EXH}}

Since $b_k = \E_{\theta_k}[X^*]$ and $\E[\hat\Lambda(T)] = 1$, we have
\begin{align*}
    \E[X^H]&= \left(\frac{\alpha-\beta}{\alpha}\right)
    \sum_{k=1}^m \left(b_k \E[\Q_T(\theta_k)] - b_k \Q_0(\theta_k) \E_{\theta_k}[\hat\Lambda(T)]\right)\\
    &= \left(\frac{\alpha-\beta}{\alpha}\right)
    \left(\sum_{k=1}^m b_k \Q_0(\theta_k) - \E\left[\E_{\theta_k}[X^*] \, \E_{\theta_k}[\hat\Lambda(T)]\right]\right)\\
    &= \left(\frac{\alpha-\beta}{\alpha}\right)
    \left(\E[X^*] \E[\hat\Lambda(T)] - \E\left[\E_{\Theta}[X^*] \, \E_{\Theta}[\hat\Lambda(T)]\right]\right).
\end{align*}
The second equality makes use of the martingale property of $\{\Q_t(\theta_k)\}_{0 \leq t \leq T}$.

\subsection{Proof of Lemma \ref{lem:bx}}

This proof calls upon the lemma presented below. Throughout this section, $H=L^2(q_0)$ and $T_2$ denotes the integral operator defined by $T_2 \phi = \int_\D K_2(x,y) \phi(y) \,\d y$ for $\phi \in H$.

\begin{lem}\label{lem:compact}
    $T_2$ is continuous and compact on $H$.
\end{lem}

\begin{proof}
    Recall that $T:H \mapsto H$ is defined by 
    \begin{align*}
        (T\phi)(x) &= \int_{\R^d} K_2(x,y) \, \phi(y) \, \d y\\
        &\int_{\R^d} \frac{K_2(x,y)}{q_0(y)} \, \phi(y) \, q_0(y) \, \d y,
    \end{align*}
    where $\tilde K_2(x,y) = \frac{K_2(x,y)}{q_0(y)}$ is the $L^2(q_0)$ kernel of $T_2$.
    Under Assumption \ref{assum:q0}, $\tilde K_2(x,y)$ is a Hilbert-Schmidt kernel and thus $T_2$ is a Hilbert-Schmidt operator in $L^2(q_0)$; continuity and compactness follow directly (see Chapter XI of \cite{dunford1963linear}). 
    \qed
\end{proof}




    
    

We can now prove Lemma \ref{lem:bx}.
We first remark that since $K_1$ is only a function of $y$, $\int_{\D} b(y) K_1(y) \d y$ is a constant once $b$ is fixed. 
Therefore, we write 
\begin{equation}\label{eq:Cb}
    C[b] = \int_{\D} b(y) K_1(y) \d y
\end{equation}
and $\tilde f(x) = f(x) + \gamma C[b]$, so that \eqref{eq:Fredholm} can be written as a modified Fredholm equation, that is
\begin{align}
    b(x) + \gamma \int_{\D} b(y) K_2(x,y) \d y = 
    \tilde f(x).
    \label{eq:Fredholm_modif}
\end{align}
Then, it is possible to first solve \eqref{eq:Fredholm_modif} for a fixed $C[b]$, and to find $C[b]$ as a second step.

\textbf{Step 1: Solve \eqref{eq:Fredholm_modif} in terms of $C[b]$}

We can re-write \eqref{eq:Fredholm_modif} in compact form as
\begin{equation*}
    (I+\gamma T_2)b = f(x) + \gamma C[b].
\end{equation*}

Under the assumption that $\frac{-1}{\gamma} \notin \sigma(T_2)$ and since $T_2$ is compact (see Lemma \ref{lem:compact}), by the Fredholm alternative, the resolvent $(I + \gamma T_2)$ is bijective and by the open mapping theorem, its inverse is bounded (see Section 4.3 and Theorem 10.8 of \cite{kress1989linear} for more details).
We also note that $\tilde f \in H$ since $\E[\hat\Lambda^2(T)]<\infty$.
Therefore, $b = (I + \gamma T_2)^{-1} \tilde f$ exists and is unique.

Using the definition of $\tilde f$ gives
\begin{align}
    b &= (I + \gamma T_2)^{-1} f + \gamma C[b](I + \gamma T_2)^{-1} \mathbf{1}\nonumber\\
    &= \psi_0 + \gamma C[b] \psi_1, \label{eq:bx_psi}
\end{align}
where $\mathbf{1}$ is the constant function 1, $\psi_0 \coloneq (I + \gamma T_2)^{-1} f$ and $\psi_1 \coloneq (I + \gamma T_2)^{-1} \mathbf{1}$.


Since $\int_{\R} K_2(x,y) \d y = 1$, we can solve for $\psi_1$ explicitly. Indeed, we have
\begin{align*}
    1 = \left(1 +\gamma \int_{\D} K_2(x,y) \d y \right) \psi_1(x)
    = (1+\gamma) \psi_1(x),
\end{align*}
so that 
\begin{align}\label{eq:psi1}
    \psi_1(x) = (1 + \gamma)^{-1}
\end{align}
for all $x \in \R$.

\textbf{Step 2: Solve for C[b]}

Using \eqref{eq:bx_psi} in \eqref{eq:Cb} yields
\begin{align*}
    C[b] &= \int_{\D} K_1(y) \psi_0(y) \d y + \gamma C[b] \int_{\D} K_1(y) \psi_1(y) \d y
\end{align*}
with 
\begin{align}\label{eq:B_explicit}
    \int_{\D} K_1(y) \psi_1(y) \d y = (1+\gamma)^{-1} \int_{\D} \E_y[\hat \Lambda]q_0(y) \d y 
    = (1+\gamma)^{-1}.
\end{align}
Thus, 
\begin{align*}
    C[b] = (1+\gamma) \int_{\D} K_1(y) \psi_0(y) \d y,
\end{align*}
and since $\gamma > 0$, this solution is valid for all possible values of $\gamma$.

The final form for $b(x)$ follows from \eqref{eq:bx_psi}, \eqref{eq:psi1} and \eqref{eq:B_explicit}.

Finally, to show that $T_2$ is self-adjoint on $H$, we rely on the definition of $q_T$ as the posterior distribution of $\Theta$ given $\F^Y_T$, rather than on its explicit expression \eqref{eq:filter}.

We denote by $p_{Y(T)|\Theta}$ the distribution of $Y(T)$ conditional on $\Theta$ and by $p_{Y(T)}$, its unconditional distribution. 
Using Bayes rule, we observe that
\begin{align*}
    K(x,y) \, q_0(x) &= \E_x[q_T(Y,y)] \, q_0(x) \\
    &= \int_{\R^d} \frac{p_{Y(T)|\Theta}(s|y) \,q_0(y)}{p_{Y(T)}(s)} \, p_{Y(T)|\Theta}(s|x) \, \d s \, q_0(x)\\
    &= \int_{\R^d} \frac{p_{Y(T)|\Theta}(s|x) \,q_0(x)}{p_{Y(T)}(s)} \, p_{Y(T)|\Theta}(s|y) \, \d s \, q_0(y)\\
    &= K(y,x) \, q_0(y).
\end{align*}
    
It follows that for $\phi$ and $\psi \in H$, 
\begin{align*}
    \langle T_2 \phi, \psi\rangle &= \int_{\R^d} \int_{\R^d} K(x,y) \, \phi(y) \, \d y \, \psi(x) \, q_0(x) \, \d x\\  
    &= \int_{\R^d} \int_{\R^d} K(y,x) \, \psi(x) \, \d x \, \phi(y) \, q_0(y) \, \d y = \langle \phi, T_2 \psi\rangle.
\end{align*}

\subsection{Proof of Theorem \ref{theo:Ambiguity_cts}}




We again skip the time index $T$ and the strategy $\pi$ in the notation of $X^\pi(T)$ and $\hat\Lambda(T)$.
Again, since the market is complete, we use the dual approach and solve first for the optimal terminal. 
Then, the inner optimization problem is 
\begin{align*}
    \sup_{X \in \mathcal \F_T^Y} \E[X]+\beta \E[(X-\bar b)^2] + (\alpha-\beta) \E[(X-b(\Theta))^2], \qquad \text{s.t. } \E[\hat \Lambda X] = x_0
\end{align*}
for fixed $\bar b \in \R$ and $b \in L^2(\D)$.

Since we optimize over $\F_T^Y$ measurable random variables, we follow the same idea as in the discrete case and write
\begin{align*}
\E[(X-b(\Theta))^2] 
&= \int_\D \E_x[(X-b(x))^2] q_0(x) \d x\\
&= \int_\D \E[(X-b(x))^2 q_T(Y,x)] \d x\\
&= \E\left[\int_\D (X-b(x))^2 q_T(Y,x) \d x\right],
\end{align*}
where the last equality is obtained using Fubini's theorem, since the integrand is strictly positive.
Then, the Lagrangian associated with this problem is
\begin{align*}
    L(X, \Lambda) = \E\left[X + \beta(X-\bar b)^2 + (\alpha-\beta) \int_\D (X-b(x))^2 q_T(Y,x) \,\d x + \lambda(\hat \Lambda X - x_0)\right].
\end{align*}
A first order condition yields
\begin{align*}
    X^*(\lambda,\bar b, b) = \frac 1{2\alpha} \left(2\beta \bar b - 1-\lambda \hat\Lambda + 2(\alpha - \beta) \int_\D b(x) q_T(Y,x) \,\d x\right).
\end{align*}

The next step is equivalent to solving for the Lagrange multiplier $\lambda$ and solving the outer optimization. We consider the conditions
\begin{enumerate}[(i)]
    \item\label{cond:budget} $\E[\hat \Lambda X] = x_0$;
    \item\label{cond:b} $\E[X] = \bar b$;
    \item\label{cond:bx} $\E_x[X] = b(x)$ for all $x \in \D$.
\end{enumerate}
From conditions \eqref{cond:budget} and \eqref{cond:b}, we obtain that $\lambda = -1$ and 
\begin{align*}
 \bar b = \frac{1}{2\beta} \left(2\alpha x_0 - \Var(\hat\Lambda) - 2(\alpha - \beta)\int_\D b(x) E_x[\hat\Lambda] q_0(x) \,\d x\right),
\end{align*}
so that the optimizer is given by
\begin{align*}
    X^* = x_0 - \frac 1{2\alpha} \left(\E[\hat\Lambda^2] - \hat\Lambda\right) + \frac{\alpha-\beta}{\alpha} \int_\D b(x) \left(q_T(Y,x) - \E_x[\hat \Lambda] q_0(x)\right) \,\d x.
\end{align*}

From condition \eqref{cond:bx}, we get that $b$ solves \eqref{eq:Fredholm}, and it follows from Lemma \ref{lem:bx} that under the condition of this theorem, \eqref{eq:bx} is a solution.

In order to determine the optimal investment strategy we can again proceed as in the discrete case. In particular we have to derive the optimal wealth process under $\tilde \P$ and  determine the strategy $\pi^*$ from identifying
\begin{align*}
    \hat X^*_t = \tilde\E[X^*|\F_t^Y] = x_0+ \int_0^t \pi_s^* \sigma dY_s.
\end{align*}
Using the preceding expression  we obtain
\begin{align*}
   \tilde\E[X^*|\F_t^Y] &=  x_0 - \frac{1}{2\alpha} \Big(\E[\hat\Lambda^2]-\tilde\E[\hat\Lambda|\F_t^Y]\Big)  - \frac{\alpha-\beta}{\alpha} \int_{\D} b(x) \Big(\tilde\E[q_T(Y,x)|\F_t^Y]-\E_{x}[\hat\Lambda]q_0(x)\Big)dx .
\end{align*}
The expression $ \tilde\E[\hat\Lambda|\F_t^Y]$ is as before  and we have
\begin{align*}
    \tilde\E[q_T(Y,x)|\F_t^Y] &= q_0(x)  \int_{\R^d} \frac{\exp(x'(Y(t)+z))-\frac12 \|x\|^2t)}{F(T,Y(t)+z)} \varphi_{T-t}(z) dz
\end{align*}
Inserting these expressions  and using It\^{o}'s lemma to derive the term in $dY_s$  we end up with the stated strategy. Note that this strategy also satisfies $ \int_0^T \E\|\pi^*(t)\|^2 dt< \infty$ due to the properties of the normal distribution.

\subsection{Proof of Theorem \ref{thm:Fredhol_sol_Gaussian}}
Step 1:
    First, it is possible to rewrite the integral equation as follows. Note that it can be checked that
    \begin{itemize}
        \item[(i)] $L=l(0)$
        \item[(ii)] $K_1(y)=\varphi_0(y) l(y) = q(0,y).$
    \end{itemize}
    Hence, $K(x,y)=q(x,y)-q(0,y).$ In what follows define the operator
    $$ (Qh)(x):= \int_\R q(x,y)h(y)dy.$$
Thus, the integral equation can be written as
\begin{align}
    b(x)= f(x) +\lambda [(Qb)(x)-(Qb)(0)].
\end{align}
Further, $b(0)=f(0)=x_0.$ Thus, when we let $b(x)=x_0+h(x),$ with $ h(0)=0$ then the remaining function solves
\begin{align}\label{eq:h-fixedpoint}
    h(x)= \frac{l(x)-l(0)}{2\alpha} + \lambda [(Qh)(x)-(Qh)(0)].
\end{align}

   Step 2: The solution of this  equation can be represented with the help of an Hermite polynomial expansion since $l$ can be written as
   \begin{align}\label{eq:lHermite}
       l(x) &= \sum_{n=0}^\infty c_n H_n\Big(\frac{x-\theta}{v}\Big)
   \end{align}
   with
    \begin{equation}\label{eq:cn_explicit}
  c_n=
  \sum_{j=0}^{\lfloor n/2\rfloor}
  \frac{(-1)^j\rho^{2j}\eta^{\,n-2j}}
       {2^j j!(n-2j)!},\quad n\in\mathbb{N}_0.
\end{equation}
In order to see this, let $z=(x-\theta)/v$ and note that
\begin{align*}
    l(\theta +vz)=  \frac{1}{\sqrt{1-\rho^2}}
   \exp\!\left(-\frac{\rho^2z^2-2\eta z+\eta^2}{2(1-\rho^2)}\right) \stackrel{!}{=} \sum_{n=0}^\infty c_n H_n(z) .
\end{align*}
Let $\varphi_{0,1}$ be the density of the standard normal. From the orthogonality of the Hermite polynomials we obtain for $k\in\mathbb{N}_0$
\begin{align*}
    \int_\R \Big( \sum_{n=0}^\infty c_n H_n(z) \Big) H_k(z) \varphi_{0,1}(z)dz= c_k k!.
\end{align*}
Multiplying with $t^k/k!$ and summing from $k=0$ to $\infty$ yields:
\begin{align*}
    &\int_\R   l(\theta +vz) \sum_{k=0}^\infty \frac{t^k}{k!} H_k(z) \varphi_{0,1}(z)dz = \sum_{k=0}^\infty c_k t^k.
   \end{align*}
Using \eqref{eq:Hermite_gen}  yields
    \begin{align}
   & \int_\R   l(\theta +vz) \exp\Big(tz-\frac{t^2}{2}\Big) \varphi_{0,1}(z)dz
   = \exp\Big( t\eta -\frac12 t^2\rho^2\Big)
    =\sum_{k=0}^\infty c_k t^k.
    \end{align}
Using the Taylor series for the exponential and equating the coefficients of the power series yields  \eqref{eq:cn_explicit}.

Step 3: We investigate what happens when we apply the $Q$-operator to a Hermite polynomial. We obtain for $n\in\mathbb{N}_0$
\begin{align}\label{eq:QtoH}
    \left(Q H_n \left(\frac{\cdot -\theta}{v}\right)\right)(x)= \rho^n H_n\Big(\frac{x-\theta}{v}\Big).
\end{align}
The proof is as follows: Using the definition of the operator this is equivalent to
\begin{align*}
    \int_\R q(x,y) H_n\Big(\frac{y-\theta}{v}\Big)dy = \rho^n H_n\Big(\frac{x-\theta}{v}\Big).
\end{align*}
Multiplying both sides by $t^n/n!$, summing over $n\in\mathbb{N}_0$ and using again \eqref{eq:cn_explicit} yields the equivalent formulation
\begin{align*}
    \int_\R q(x,y) \exp\Big(t\Big( \frac{y-\theta}{v}\Big)-\frac{t^2}{2}\Big)dy = \exp\Big(t\rho \Big( \frac{x-\theta}{v}\Big)-\frac{(t\rho)^2}{2}\Big)
\end{align*}
which can be checked to be true.

Step 4: Define for $n\in\mathbb{N}:$
\begin{align*}
    \psi_n(x) := H_n\Big(\frac{x-\theta}{v}\Big)-H_n\Big(-\frac{\theta}{v}\Big)
\end{align*}
For $h(x)=b(x)-x_0$ we consider the Ansatz: 
\begin{align*}
    h(x) = \sum_{n=1}^\infty b_n \psi_n(x).
\end{align*}
Inserting this into equation \eqref{eq:h-fixedpoint} and using \eqref{eq:lHermite} as well as \eqref{eq:QtoH} yields
\begin{align*}
    \sum_{n=1}^\infty b_n \psi_n(x) &= \frac{1}{2\alpha} \sum_{n=1}^\infty c_n \psi_n(x) + \lambda \Big[\sum_{n=1}^\infty b_n \rho^n  H_n\Big(\frac{x-\theta}{v}\Big) - \sum_{n=1}^\infty  b_n \rho^n H_n\Big(-\frac{\theta}{v}\Big)\Big]\\
    &= \frac{1}{2\alpha} \sum_{n=1}^\infty c_n \psi_n(x) + \lambda \sum_{n=1}^\infty b_n \rho^n  \psi_n(x).
\end{align*}
Comparing the coefficients yields: $(1-\lambda \rho^n)b_n=c_n/2\alpha.$ The case $\lambda \rho^n =1$ cannot occur in our case, since $\lambda <0.$ This concludes the proof.

\subsection{Proof of Theorem \ref{theo:X_Gaussian}}

Recall that 
\begin{equation*}
    \begin{split}
    X^* &=\; x_0-\frac{1}{2\alpha}\Big(\E[\hat\Lambda^2(T)] - \hat\Lambda(T)\Big) \\ 
    & + \Big(\frac{\alpha-\beta}{\alpha}\Big) \int_{\D} b(x) \left(q_T(Y,x) - \E_x[\hat\Lambda(T)] q_0(x) \right) \d x,
    \end{split}
\end{equation*}
and we use 
\begin{equation*}
    b(x) = x_0+\frac{1}{2\alpha} \sum_{n=1}^\infty \frac{c_n}{1-\lambda \rho^n}\Big[ H_n\Big(\frac{x-\theta}{v}\Big)-  H_n\Big(-\frac{\theta}{v}\Big)\Big].
\end{equation*}

Using the notation of the proof of Theorem \ref{thm:Fredhol_sol_Gaussian}, the terms that do not depend on $Y(1)$ are
\begin{align*}
    C_1&=x_0 - \frac{L}{2\alpha} - \lambda \int_\R q(0,x) b(x) \,\d x\\
    &=x_0 (1-\lambda) - \frac{L}{2\alpha} + \frac{\lambda}{2\alpha} \sum_{n=1}^\infty \frac{c_n}{1-\lambda \rho^n} H_n \left(-\frac{\theta}{v}\right)\\
    &\quad - \frac{\lambda}{2\alpha} \sum_{n=1}^\infty \frac{c_n}{1-\lambda \rho^n} \left(QH_n\left(\frac{\cdot -\theta}{v}\right)\right)(0)\\
    &= x_0 (1-\lambda) - \frac{L}{2\alpha} + \frac{\lambda}{2\alpha} \sum_{n=1}^\infty \frac{c_n}{1-\lambda \rho^n} \left((1-\rho^n) H_n \left(-\frac{\theta}{v}\right)\right),
\end{align*}
where \eqref{eq:QtoH} is used to obtain the last equality.

The terms that depend on $Y(1)$ are 
\begin{align*}
    &\frac 1{2\alpha} \hat\Lambda(1) + \lambda\int_\R b(x) q_1(Y,x) \,\d x\nonumber \\
    \begin{split}
     &=\frac 1{2\alpha} \hat\Lambda(1) + \lambda\left(x_0 - \frac 1{2\alpha} \sum_{n=1}^\infty \frac{c_n}{1-\lambda \rho^n} H_n\left(\frac{-\theta}{v}\right)\right) \int_\R q_1(Y,x) \,\d x \\
    &\quad + \frac{\lambda}{2\alpha}\int_\R \sum_{n=1}^\infty \frac{c_n}{1-\lambda \rho^n} H_n\left(\frac{x-\theta}{v}\right) q_1(Y,x) \, \d x.
    \end{split}
\end{align*}
Careful integration yields
\begin{align*}
    \int_\R q_1(Y,x) \,\d x = \frac{\hat\Lambda(1)}{\sqrt{v^2+1}}\exp\left(\frac{Y(1)^2}{2} - \frac{(Y(1)-\theta)^2}{2(v^2+1)}\right)
\end{align*}
and 
\begin{align*}
    &\frac{\lambda}{2\alpha}\int_\R \sum_{n=1}^\infty \frac{c_n}{1-\lambda \rho^n} H_n\left(\frac{x-\theta}{v}\right) q_1(Y,x) \, \d x\\
    &= \frac{\lambda}{2\alpha} \hat\Lambda(1) \int_\R \sum_{n=1}^\infty \frac{c_n}{1-\lambda \rho^n} H_n\left(\frac{x-\theta}{v}\right) \frac 1{\sqrt{2\pi v^2}} e^{xY(1) - \frac{x^2}{2} - \frac{(x-\theta)^2}{2v^2}} \, \d x\\ 
    &= \frac{\lambda}{2\alpha} \frac{\hat\Lambda(1)}{\sqrt{v^2+1}}\exp\left(\frac{Y(1)^2}{2} - \frac{(Y(1)-\theta)^2}{2(v^2+1)}\right)\sum_{n=1}^{\infty}\frac{c_n\,\rho^{n/2}}{1-\lambda\rho^n}H_n\!\left(\frac{(Y(1)-\theta)\sqrt{\rho}}{v}\right),
\end{align*}
and the solution follows.


\bibliography{MVvsUtility}

\end{document}